\documentclass[12pt, a4paper,reqno]{amsart}
\usepackage[margin=1in]{geometry}
\usepackage{amssymb}
\usepackage{enumerate}
\usepackage{hyperref}
\usepackage{doi,url}
\usepackage{dsfont}
\usepackage{mathabx}
\usepackage[normalem]{ulem}
\usepackage{cancel}

\usepackage[dvipsnames]{xcolor}
\usepackage{tikz}
\usepackage{pgfplots}
\usepgfplotslibrary{fillbetween}
\pgfplotsset{compat=1.16,width=10cm}
 \usepackage{caption}
\usepackage{float}
\usetikzlibrary{patterns}

\definecolor{purple}{rgb}{.5,0,1}
\definecolor{orange}{rgb}{1,.5,0}
\definecolor{pink}{rgb}{1,0,.5}

\usepackage{fancyvrb}
\usepackage{amsmath,amsthm}
\usepackage{enumitem}
\usepackage{thmtools}
\usepackage[style=apalike,style=numeric,maxbibnames=99,giveninits=true,citestyle=numeric,url=true,doi=true,
backend=biber
]
{biblatex}
\renewbibmacro{in:}{}
\AtEveryBibitem{\clearfield{number}}
\DeclareFieldFormat*{volume}{\mkbibbold{#1}}
\DeclareFieldFormat[article]{pages}{#1}

\numberwithin{equation}{section}
\newtheorem{theorem}{Theorem}[section]

\newtheorem{lemma}[theorem]{Lemma}

\newtheorem{remark}[theorem]{Remark}

\newtheorem*{theorem*}{Theorem}
\newtheorem*{lemma*}{Lemma}
\newtheorem*{remark*}{Remark}

\DeclareMathOperator{\supp}{supp}

\DeclareMathOperator{\tr}{tr}
\DeclareMathOperator{\Ran}{Ran}

\DeclareMathOperator{\sgn}{sgn}
\DeclareMathOperator{\dist}{dist}

\DeclareMathOperator{\Rea}{Re}
\DeclareMathOperator{\Ima}{Im}

\renewcommand\H{\mathcal{H}}

\newcommand\R{\mathbb R}
\newcommand\N{\mathbb N}
\newcommand\C{\mathbb C}
\newcommand\Z{\mathbb Z}

\newcommand\D{\mathcal{D}}

\newcommand\cT{\mathcal{T}}

\newcommand\cW{\mathcal{W}}

\newcommand\x{\mathbf{x}}
\newcommand\y{\mathbf{y}}
\renewcommand\u{\mathbf{u}}

\renewcommand\v{\mathbf{v}}

\newcommand\e{\mathrm{e}}

\renewcommand\P{\mathbb P}
\newcommand\E{\mathbb E}
\newcommand\cE{\mathcal{E}}

\newcommand\cN{\mathcal{N}}

\newcommand\cH{\mathcal{H}}

\newcommand\cS{\mathcal{S}}

\newcommand\cQ{\mathcal{Q}}

\newcommand\cP{\mathcal{P}}

\newcommand\eps{\varepsilon}
\newcommand\vphi{\varphi}

\renewcommand{\d}{\mathrm{d}}

\newcommand{\pr}{\prime}

\newcommand\what{\widehat}
\newcommand\wtilde{\widetilde}

\newcommand\beq{\begin{equation}}
\newcommand\eeq{\end{equation}}
\newcommand\be{\begin{equation}\begin{aligned}}
\newcommand\ee{\end{aligned}\end{equation}}

\newcommand{\abs}[1]{\left\lvert #1 \right\rvert}
\newcommand{\norm}[1]{\left\lVert #1 \right\rVert}
\newcommand{\scal}[1]{\left\langle #1 \right\rangle}
\newcommand{\set}[1]{\left\{ #1 \right\}}
\newcommand{\pa}[1]{\left( #1 \right)}

\newcommand{\fl}[1]{\left\lfloor #1 \right\rfloor}

\newcommand{\cl}[1]{\lceil #1 \rceil}

\newcommand\La{\Lambda}

\newcommand{\eq}[1]{\eqref{#1}}

\newcommand{\up}[1]{^{\left(#1\right)}}

\newcommand{\qtx}[1]{\quad\text{#1}\quad}
\newcommand{\mqtx}[1]{\; \ \text{#1}\; \  }
\newcommand{\sqtx}[1]{\;\text{#1}\;}

\newcommand{\bD}{\boldsymbol{\Delta}}

\newcommand{\tfd}{\pa{1- \tfrac{1}{\Delta}}}
\newcommand{\fd}{1- \frac{1}{\Delta}}

\newcommand{\nfd}{\pa{1- \frac{1}{\Delta}}}

\newcommand{\prr}{{\pr\pr}}

\newcommand{\clq}{\cl{q}}

\newcommand{\ident}{\mathds{1}}

\begin{document}

\title[Localization in the  random XXZ  spin chain]{Localization phenomena in the  random XXZ  spin chain}
\author{Alexander Elgart}
\address[A. Elgart]{Department of Mathematics; Virginia Tech; Blacksburg, VA, 24061-1026, USA}
 \email{aelgart@vt.edu}

\author{Abel Klein}
\address[A. Klein]{University of California, Irvine;
Department of Mathematics;
Irvine, CA 92697-3875,  USA}
 \email{aklein@uci.edu}


\begin{abstract}

It is shown that the infinite random Heisenberg XXZ spin-$\frac12$ chain exhibits localization phenomena, such as spectral, eigenstate, and weak dynamical localization, in an arbitrary (but fixed)  energy interval in a non-trivial  region of the parameter space. This region depends only on the energy interval and includes weak interaction and strong disorder regimes. 
The crucial step  in the argument is  a proof that if the  Green functions for the associated finite  systems Hamiltonians exhibit certain (volume-dependent) decay properties in a fixed energy interval, then the infinite volume Green function  decays  in the same interval as well.  The pertinent finite systems  decay properties  for the random XXZ spin chain  had been previously verified by the authors.
 \end{abstract}



\setcounter{tocdepth}{1}
\maketitle

\tableofcontents

\section{Introduction}\label{secmodel}

\subsection{The  model and    (informal) main result}

The simplest non-trivial quantum system is a single spin-$\frac 12$, characterized by a two-dimensional complex state space $\C^2$, spanned by  two orthonormal  vectors called qubits: the spin-up  $\uparrow\rangle = \begin{pmatrix} 1  \\ 0  \end{pmatrix} $ and the  spin-down  $\downarrow\rangle = \begin{pmatrix} 0  \\ 1 \end{pmatrix} $ states. The self-adjoint operators on this space are real linear combinations of the identity $\mathds{1}_{\C^2}$ and the three Pauli matrices,
\[\sigma^x=\begin{pmatrix}0& 1  \\1& 0  \end{pmatrix},\quad \sigma^y=\begin{pmatrix}0& -i  \\i& 0  \end{pmatrix},\quad \sigma^z=\begin{pmatrix}1& 0  \\0& -1  \end{pmatrix}.
\]

Spin chains are arrays of  spins indexed by subsets $\La\subset\Z$. If $\La$ is finite, the corresponding state space is the tensor product Hilbert space $\mathcal H_\Lambda=\otimes_{i\in \Lambda} \cH_i$, where each $\cH_i$ is a copy of $\C^2$. For infinite $\La$, we let  $\cH_{\La,0}$ be the vector subspace of $\bigotimes_{i\in \La}\cH_i$ spanned by    tensor products of the form  $\bigotimes_{i\in \Z}\vphi_i$,  $\vphi_i\in \set{\uparrow\rangle_i,\downarrow\rangle_i}$, with a finite number of spin-downs, equipped with the tensor product space inner product,  and let $\cH_\La$ be its Hilbert space completion. A single spin operator $\sigma^\sharp$ acting on the $i$-th spin is lifted to  $\mathcal H_\Lambda$  by identifying it with  $\sigma^\sharp \otimes \mathds{1}_{\La\setminus \set{i}}$, where $ \mathds{1}_{\La\setminus \set{i}}$ denotes the identity operator on $\cH_{\La\setminus \set{i}}$, i.e.,   it acts non-trivially only in the tensor product's $i$-th component. To stress the $i$-th dependence, we will denote this single spin operator by $\sigma^\sharp_i$. More generally,  if  $S\subset \La$, and $A_S$ is an operator on  $\cH_S$, we often  identify it with its natural embedding on $\cH_\La$, namely with  $A_S \otimes \mathds{1}_{\La\setminus S}$.

The original motivation to study quantum spin systems goes back to the 1920s when their usefulness in explaining ferromagnetism was realized by Lenz, Ising, Dirac, and Heisenberg, among others. They are now playing a role in explaining various phenomena across physics, computer science, chemistry, and biology. The rich structure associated with these systems is related to their complexity: While a single spin has a very simple state space, the dimensionality of $n$ spins grows exponentially fast with $n$.  Even for modestly sized spin systems where $n$ ranges into the dozens, the computational cost of their numerical analysis is prohibitive. This problem is colloquially known in physics as the curse of dimensionality and is the main cause of our very limited theoretical understanding of such models, especially of their thermodynamic limit $\La\to\Z$.

In this work, we study spectral and dynamical properties of the random XXZ quantum spin-$\frac 12$  chain. The   random Hamiltonian $H^\Lambda =H_\omega^\Lambda$ on $\cH_\La$  is given by\footnote{ Our definition of $H^\Lambda$ incorporates a choice of  boundary condition if  $\La\ne \Z$ .}   
\be\label{eq:H}
H^\Lambda=  -\tfrac{1}{2\Delta}\sum_{\set{i,i+1}\subset \Lambda} \pa{\sigma_i^+\sigma_{i+1}^-+\sigma_i^-\sigma_{i+1}^+}  + \sum_{i\in \Lambda} \mathcal{N}_i-\sum_{\set{i,i+1}\subset \Lambda} {\cN_i\cN_{i+1}} +\lambda\sum_{i\in \Lambda} \omega_i \mathcal{N}_i,
\ee 
where   $\sigma^\pm=\frac 12( \sigma^x \pm i \sigma^y)$ are  the ladder operators  and $\mathcal{N} = \tfrac{1}{2} (\mathds{1}_{\C^2}-\sigma^z)$ is known as the number operator. The constant  $\Delta$  is the anisotropy   parameter;   we assume $\Delta>1$  (the Ising phase). The parameter $\lambda>0$ determines the strength  of a random transversal field $\sum_{i\in \Lambda} \omega_i \mathcal{N}_i$, where $\omega = \set{\omega_i}_{i\in\Z}$  is a family of  random variables.  \emph{Throughout this work  we assume that  $\set{\omega_i}_{i\in\Z}$  are independent, identically distributed random variables, whose  common probability 
distribution $\mu$ is absolutely continuous with a bounded density and satisfies $\set{0,1}\subset \supp \mu\subset[0,1]$.}

$H^\La$ is a well defined  positive bounded self-adjoint operator for finite sets $\La$.  For infinite sets $H^\La$ is understood as an unbounded positive self-adjoint operator on $\mathcal H_\Lambda$.   Alternatively, one can exploit the fact that  $H^\La$ commutes with the  total  number of particles operator $\cN^\La= \sum_{i\in \La} \cN_i$  to represent it as a direct sum of bounded Hamiltonians of systems with a fixed number $N$ of particles. The corresponding $N$-particles Hamiltonian $H^\La_N$  can be seen as a random Schr\"odinger operator on a certain subgraph of $\La^N$.

The free ($\lambda=0$) XXZ system is a special variant of  the famous Heisenberg model. For $\La=\Z$, its  spectrum can be determined using the Bethe ansatz, a method introduced by Bethe in 1931.
 Its ground state energy $0$ is separated by a gap of size  $\fd$ from the rest of the spectrum, which is expected to be absolutely continuous due to the translation invariance of the underlying Hamiltonian. This feature has been verified for some energy intervals \cite{NaSpSt,FiSt}.

 Starting from the first decade of the new millennium, the randomized version of this operator ($\lambda>0$) has been proposed as a prototypical model for the  study of  many-body localization (MBL) phenomena  in solid state physics.  The initial investigations (of numerical and heuristic nature) in the physics community seemed to indicate that for large values of $\lambda$  a completely different system's behavior emerges: The spectrum becomes pure point almost surely and  the system's dynamics changes drastically, 
with thermalization not occurring even in the asymptotic limit of infinite system size and evolution time. This behavior is called the {\it MBL phase}. While the localization phenomenon for single-particle systems is well understood, and indeed persists in infinite volume systems for all times, it is an open question in physics whether the MBL phase does occur. It is also not clear what its precise characterization is, see \cite{sierant2024} for the recent review (other physics reviews on this topic include \cite{NandHuse,AL,abanin2019}).

  For single-particle systems, one usually distinguishes three types of localization: spectral, eigenstate, and dynamical localization. We will now briefly describe these forms of localization to put our results for the random XXZ model in that context; a more detailed description can be found in Appendix \ref{sec:nom}. We will adopt the same nomenclature for the many-body case, though we caution the reader that  single-particle localization and MBL describe different phenomena.
 
 {\it Spectral localization}  for a random operator $H_\omega$ in a prescribed energy interval $I$ manifests itself as pure point spectrum in  $I$, almost surely. 
 This type of localization is informative only for infinite systems, but it is necessary for formulation of the different types of localization in such models.

{\it Eigenstate localization} is described at the level of the eigenvectors for $H_\omega$, and reflects their spatial confinement. A strong form of eigenstate localization is exponential decay of the {\it eigencorrelator}, see \eqref{eq:eigcor} below for its definition.

{\it Dynamical localization} is the non-spreading of initially localized wave packets (in the Schr\"odinger picture) or of local observables (in the Heisenberg picture), in the course of the time evolution generated by $H_\omega$.  
Eigenstate localization alone is typically not sufficient to guarantee dynamical localization. The decay of the eigencorrelator can be seen as a weak (in the operator-topological sense) form of dynamical localization. Since local observables for a single-particle system are either compact (in the discrete case) or relatively compact (in the continuum case),  weak dynamical localization  implies dynamical localization for such systems. As a result, for single-particle models the decay of the eigencorrelator is essentially synonymous with dynamical localization. This is no longer the case for many-body systems, where local observables are full rank operators.

The existing mathematical results for few-particles systems (e.g.,  \cite{CS,AW2,KlN1,KlN2}) show that for sufficiently large parameters $\Delta$ and $\lambda$ the infinite volume Hamiltonian $H_N$ obtained by restricting the full Hamiltonian to the $N$ particle sector is spectrally, eigenstate,  and weakly dynamically localized for all energies, provided  $N\le N_0(\Delta,\lambda)<\infty$. These methods could not be significantly improved by considering energies in a fixed interval  $[0,E_0]$.

Our main result,  stated informally below,    shows that the infinite volume random XXZ model is spectrally, eigenstate, and weakly dynamically localized in a fixed energy interval $[0,E_0]$, {\it uniformly in $N$}, as long as $\lambda \Delta^2$ is sufficiently large.   This regime is sometimes referred to in physics as  {\it zero temperature MBL}, or,    more precisely, {\it low density MBL}. We sketch both the few-particle and zero temperature localization regimes in Figure \ref{fig1}.  The precise mathematical formulation is  given in Theorem \ref{thm:ppspec1}.

 \begin{theorem}[Informal formulation]\label{thm:ppspec}  
 Let $H^\Z$ be the random XXZ Hamiltonian on $\cH_\Z$ with  parameters $\Delta>1$and $\lambda >0$. Fix the energy    interval $I(E_0)=[0,E_0]$, where $E_0>0$.   Then, if  $\lambda \Delta^2$ is sufficiently large, we have:
\begin{enumerate}
    \item   Spectral   and eigenstate  localization in the interval  $I(E_0)$:  The spectrum of $H^\Z$ in $I(E_0)$ is almost surely pure point, and the corresponding eigenvectors  in  $I(E_0)$ decay exponentially fast away from their localization centers (in a suitable sense).
\item  Weak dynamical localization in the interval $I(E_0)$:  The expectation of the absolute value of the   matrix elements of $\chi_{I(E_0)}(H^\Z)\e^{itH^\Z}$ decays exponentially fast (in a suitable sense), uniformly in $t\in \R$.
\end{enumerate}
\end{theorem} 

 \begin{figure}
\begin{center}
\begin{tikzpicture}[,scale=0.9]
 \begin{axis}[axis lines=middle,axis on top,xlabel=$N$,ylabel=$E$,
 xmin=-0.5,xmax=10,ymin=-0.2,ymax=10, ytick={1},yticklabels={$E_0$},
 xtick={1},xticklabels={$N_0$}, 
 every axis x label/.style={at={(current axis.right of origin)},anchor=north west},
 every axis y label/.style={at={(current axis.above origin)},anchor=north east}]
  \addplot[name path=g,domain=1:10,white] {min(10,1/(x-1)+1)};
    \path[name path=axis] (axis cs:1,1) -- (axis cs:10,1);

    \addplot [
        thick,
        color=red,
        fill=red, 
        fill opacity=0.1
    ]
    fill between[
        of=g and axis,
        soft clip={domain=1:10},
    ];
\fill[pattern=north west lines, pattern color=blue] (0,0) rectangle (1,10);
\fill[pattern=north east lines, pattern color=green] (0,0) rectangle (10,1);

 \path (0.5,5) node{$A$}
 (5,5) node{$D$}
  (1.7,1.7) node{$C$}
(5,0.5) node{$B$};
 \end{axis}
\end{tikzpicture}
\end{center}
\caption{A localization cartoon for the infinite volume XXZ model in  strong disorder/weak interaction regimes. The  blue region $A$ is the few-particles localization $N\in [0,N_0]$, the  green region $B$ is the zero temperature localization $E\in[0, E_0]$  (our result). The total region of localization can be extended to include the pink sector $C$ using existing methods. The white region $D$ is currently not understood.}\label{fig1}
\end{figure}
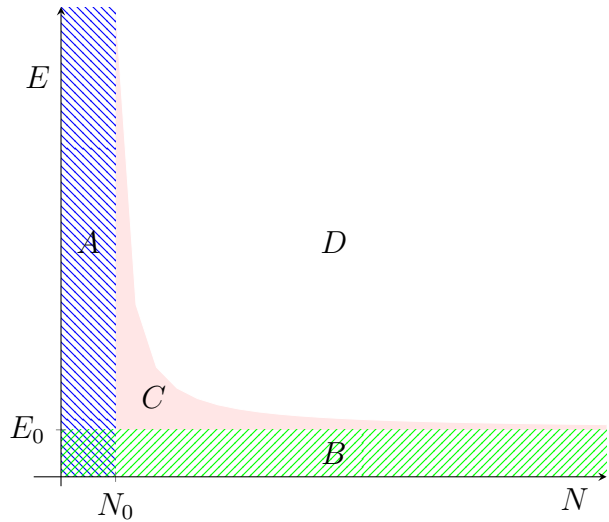

We believe these are the first results of this kind for  non-exactly solvable random spin chains.
 Previous localization results for random spin chains were restricted to either exactly solvable models  (see the review   \cite{ARNSS}, or the random XXZ  spin chain  in  the   droplet spectrum, a  special energy interval at the bottom of the spectrum.

 The dynamical behavior of a random system can be linked to two competing effects: The natural tendency of eigenfunctions to decay exponentially in the classically forbidden regions (which are numerous in say the large disorder regime) and their tendency to delocalize in the presence of resonances. For random Schr\"odinger operators on $\Z^d$, the number of resonances grows polynomially in the  system size (with degree equal to $d$), and exponential decay wins, almost surely, causing localization.  The curse of dimensionality can change this outcome already for a single particle operator on trees, e.g., on the Bethe lattice, as the number of resonances in this situation grows exponentially fast with the system size. This makes the underlying analysis quite delicate. Another (but somewhat related) adverse effect of the curse of dimensionality is related to the scarcity of randomness in the many-body setting: The number of random variables is actually the same as for the single particle model (that is, it grows only polynomially in the system size). As randomness is used to control the resonances in order to yield localization results, the new ideas were needed to control the exponentially large number of resonances in the many-body setting. We note that the scarcity of randomness has already been an issue for few-particle localization, where the number of resonances still grows only polynomially fast in the system size. Roughly speaking, the limitation of our result to the fixed energy interval is related to the fact that only in this context we can resolve both issues related to the  curse of dimensionality.

\subsection{Relation with existing research results and open problems}\label{subsec:past}
As we already mentioned, despite intensive efforts in the condensed matter physics community in the past two decades, even the very existence of the MBL phase remains a point of debate in the physics literature. On the mathematical level,   limited progress has been made in understanding this phenomenon, mostly related to the random XXZ model. As far as we are aware, Theorem \ref{thm:ppspec} is the first result that  establishes localization properties for not exactly solvable infinite spin systems in this generality.  That being said, from the physics perspective results of this kind would constitute a clear indication of an MBL phase only if the energy intervals were allowed to grow with the system size. 

 Up to a few years ago, rigorous MBL-related results  were  restricted to the class of exactly solvable models (see the review   \cite{ARNSS}). 
More recently,  for the random XXZ model, spectral and  dynamical localization in the  special energy interval $[\fd, 2\tfd)$,   called the {\it droplet spectrum},  was established in  \cite{BeW,EKS1}. These results led to the validation of other important many-body features associated with MBL in the droplet spectrum, among them the exponential clustering properties for associated eigenfunctions and non-propagation of information \cite{EKS3}, and the area law for the entanglement entropy \cite{BeW1}.

 The finite volumes bound obtained in \cite{EK22} for  a fixed energy interval provides a critical input  for the current paper,  and  were used to obtain   dynamical (rather than weak dynamical) localization type results,  expressed as slow propagation of information, for finite systems in the same regime \cite{EK24}.  However,  the bounds in \cite{EK24} depend on the volume of the system, which excludes any conclusions for an infinite system.  If the volume dependence could be suppressed, these results could be used to establish the stability of the  logarithmic light cone against generic  local perturbations \cite{TB}.
 It would be interesting to see whether the estimates developed in this paper could yield a significantly stronger version of the result established in \cite{EK24}, but it is an open question whether this volume dependence can be completely removed using these methods.

Most of the MBL attributes would be achieved if one could show the existence of  a  quasi-local unitary $U$ such that with large probability $U^*HU$ is a diagonal Hamiltonian, in the sense that it commutes with every $\sigma_i^z$ operator.  The implication is that $U^*HU$ can be represented as a sum of  weighted products $\prod_{i\in X}\sigma_i^z$  over subsets $X\in\La$ which are referred to as {\it local integrals of motions} (LIOM) in the physics literature. The LIOM representation  was first proposed in \cite{SPZA} as a possible mechanism explaining MBL, and became a popular physics  tool  for the heuristic derivation of  the majority of MBL features. In particular, the existence of such $U$  implies that one can construct an eigenbasis for $H$ consisting of vectors of the form $U\psi$, where $\psi$ is a product state.  

 For the Anderson model in the strong disorder regime, in any finite volume  one can construct a (semi-uniformly) quasi-local $U$ that diagonalizes the Hamiltonian.  (This follows from the results in  \cite{EK}.) The construction relies on eigenstate localization for all energies and the possibility to label the underlying eigenfunctions according to the spatial position of their localization centers.  Such labeling is attainable for the Anderson model due to the fact that one can show that with large probability the spectrum of $H$ is level spaced for sufficiently regular distribution of the random potential. The methods used in the present work are not sufficient to establish LIOM localization for two reasons: (a) We can only prove localization in a fixed energy interval and not on the whole spectrum of $H$; and  (b) It is not known whether the spectrum of $H$ is level spaced with large probability. It is  not clear whether it is reasonable (even from a physics perspective) to expect that such $U$ exists in the first place.

The recent preprints \cite{DGHP,Lemm}  consider weak deterministic perturbations of  diagonal (and thus exactly solvable) models. 
In \cite{DGHP}, the focus is on the one-dimensional Ising model with a random longitudinal field. The model is considered on intervals $\La$ of size $\gamma^{-c}$, where $\gamma$ is the perturbation strength and $c>0$ is a small, non-zero exponent. A key finding announced in this work is the construction of a uniformly quasi-local unitary operator $U$ such that $U^*HU$ is a diagonal operator. The authors use this to derive information about the spin chain, consistent with Many-Body Localization (MBL), specifically that the spin transport in this system is anomalous.

The paper \cite{Lemm} examines a  broad class of models on finite graphs given by $H = H_0 + \sum _{r=i}^\infty \epsilon^r V_r$. Here, $\epsilon$ is a small parameter, $V_r$ is a spatially local operator  supported on sets whose diameter is smaller than $r$, and $H_0$ is a diagonal operator formed by a sum of Pauli-Z strings. Assuming  certain non-resonant conditions on $H_0$, the authors demonstrate that the propagation of information remains slow up to spatial scales that are powers of $\ln(\epsilon^{-1})$.   In the case of   the finite volume random XXZ model, we proved in \cite{EK24} that there is  slow propagation of information across all spatial scales in  the region of localization (the finite energy interval).

 Finally, let us mention that recently there has been interest in the physics community in low-energy MBL. In particular, in \cite{YNL} it is shown that every eigenstate in an (extensive) energy interval at the bottom of the spectrum is localized for a deterministic finite system of error-correcting codes, settling the long-standing conjecture in the physics literature that such eigenstate localization is possible. Theorem 1.1 can be viewed as settling this conjecture for an infinite random system, although it applies to fixed (rather than extensive) energy intervals.

 The proof of Theorem \ref{thm:ppspec} uses the finite volume results developed in \cite{EK22} as input to obtain results for infinite systems, suppressing the volume dependence in these estimates.  In the next section we introduce the necessary technical notation, state the precise formulation of Theorem \ref{thm:ppspec}, and present  
 the technical results that  establish it.

\section{Main result and technical steps}\label{secfinvol}

 We let $\abs{S}$ stand for the cardinality of $S\subset \Z$. Given  $ \La \subset \Z$, we let \[\cP_f(\La)=\set{\x\subset  \La:  \abs{\x}<\infty}\] be the collection of  finite  subsets of $\La$,  and let \[\cP_N(\La)=\set{\x\subset  \La: \abs{\x}=N}\] for $N\in\N^0=\N\cup \set{0}$ be  the  subset of $\cP_f(\La)$ consisting of sets with cardinality $N$. We  use the notation $\cP_+(\La)=\cP_f(\La)\setminus\set{\emptyset}$.
We also  set $\abs{\x}_1=\sum_{x\in \x} \abs{x}$ and $\abs{\x}_2=\pa{\sum_{x\in \x} \abs{x}^2}^{\frac 12}$  for $\x \in   \cP_f(\Z)$. 
 
  We fix $\Delta>1$ and $ \lambda >0$.  Given  $ \La \subset \Z$,  let $H^\La$ be the Hamiltonian  given in \eq{eq:H}, and consider the canonical orthonormal basis $\Phi_\La=\set{\phi_\x}_{\x \in \cP_f(\La)}$  for $\cH_\La$, where
 \be
 \phi_\emptyset=\phi_\emptyset^\La=\otimes_{i \in \Lambda}\uparrow\rangle_i, \quad \phi_{\bf x}=\phi_{\bf x}^\La=\pa{\prod_{i\in {\bf x}}\sigma_i^-}\phi_\emptyset^\La \sqtx{for} \x \in \cP_+(\La).
 \ee
 Note that $\phi_\emptyset$, the {\it vacuum state}, is an eigenvector for $H^\La$ with the simple eigenvalue $0$.  (It is the  ground state for $H^\La$,  as we shall see later.) Note also that $\phi_\x^\La=  \phi_\x^S \otimes  \phi_\emptyset^{\La\setminus S}$ for $ \x \subset S \subset \La$.   (We will suppress the $\La$ dependence from $\phi_\x$ for ease of notation when it is clear from the context.)
$ \Phi_\La$ can be decomposed into the disjoint union
\be\label{eq:basis}
 \Phi_\La=\bigcup_{N=0}^{\abs{\La}}  \Phi_\La\up{N}, \mbox{ where }\Phi_\La\up{N}=\set{\phi_\x : \x \in \cP_N(\La)}.  
 \ee

  Given $ S\subset \Z$ finite, we set
\be\label{eq:P_+}
P_+^S&=\otimes_{i\in S}\pa{{\mathds{1}_{i}}-\cN_i} \qtx{and}  P_-^S=\mathds{1}_S-P_+^S, \qtx{if} S\ne \emptyset,\\
 P_+\up{\emptyset }& =\mathds{1}_S \qtx{and } P_-\up{\emptyset }=\mathds{1}_S- P_+\up{\emptyset }=0.
\ee
Note that $\cN_i= P_-^{\set{i}}$ for all $i\in \Z$.

 For a unit vector $\vphi \in \cH_\La$,  we denote by  $\pi_\vphi$ the orthogonal projection onto $\C \vphi$.  If $\u \in \cP_f(\La)$, we write $\pi_\u=\pi_{{\phi_\u}}$.   We have 
 \be\label{piu}
\pi_{\u}=  P_+^{\La\setminus \u}\Pi_{u\in \u} \cN_u   \qtx{for all} \u \in \cP_+(\La),
\ee

    We  denote by $R_z^\Lambda=\pa{H^\La-z}^{-1}$   the resolvent operator,  which is  well defined for    $z\in\C\setminus \R$ (and for almost all  $z\in\R$ if $\La$ is finite),  and  use the Green function notation,
\be\label{eq:Green}
G^\La_z({\bf x},{\bf y}) = \langle \phi_{\bf x},R_z^\Lambda \phi_{\bf y}\rangle \qtx{for finite} \x,\y \subset\La.
\ee
(For real valued $z$, $G^ \La_{z}({\bf x},{\bf y})$ is understood here and below as $G^ \La_{z+ i0}({\bf x},{\bf y})$.)
More generally, for an operator $T$ on $\cH_\La$ we  denote by $T({\bf x},{\bf y})$ its matrix elements with respect to the standard basis, i.e., 
 \be\label{eq:Txy}
T({\bf x},{\bf y}) = \langle \phi_{\bf x},T \phi_{\bf y}\rangle;  \qtx{note} \abs{T({\bf x},{\bf y})}=\norm{\pi_{\x}T\pi_{\y}}.
\ee
In addition, we set
\be
\psi(\x)=\scal{\phi_\x, \psi} \qtx{for all} \psi \in \cH_\La \mqtx{and} \x\in \cP_f (\La).
\ee

 We let $B(I)$   denote the collection of all bounded Borel measurable functions $f$  supported on the interval  $I$, and set $B_1(I)=\set{f\in B(I),\ \sup_{t\in I} \abs{f(t)}\le 1}$.

We equip $\Z$ with the usual graph  distance  $d_\Z(x,y)= \abs{x-y}$ for  $x,y \in \Z$.  We will  consider $\La \subset \Z$ as a subgraph of $\Z$,  and
 denote by $\dist_\La( \cdot,\cdot) $ the  graph distance in $\La$, which  can be infinite if $\La$ is not a connected subset of $\Z$.
  Given  $ S\subset \Lambda\subset \Z$ and  $p\in\N^0$,   we set
   \be
{[S]^\La_p} & = \set{x\in\La: \dist_\La\pa{x,S}\le p},\\
\partial_{ex}^\La S& = \set{x\in\Lambda:\ d_\La\pa{x,S}=1}= [S^\La]_1\setminus S,\\
 \partial_{in}^\La S& = \set{x\in\Lambda:\ \dist_\La \pa{x, \La \setminus S}=1},\\
 \partial^\La S & =   \partial_{in}^\La S \cup  \partial_{ex}^\La S.
\ee

 We also  consider the Hausdorff distance between subsets of $\La$, given by
 \be
 d_H^\La (U,V)= \max \set{\max_{u\in U} d_\La (u, V), \max_{v\in V} d_\La (v, U)   }\qtx{for}  U,V \subset \La,
 \ee
and observe that
 \be
 d_H^S (U,V) \ge d_H^\La (U,V) \qtx{if}   U,V\subset S \subset \La.
 \ee

Due to the conservation of the total magnetization in the XXZ spin chain (see the next section), for any $z\in \C$ and, more generally, for any bounded Borel measurable function $f$,  we have 
\be\label{fxy0}
 G_z^\La({\bf x},{\bf y})= f(H^{\La})({\bf x},{\bf y})=0 \qtx{for all} \x,\y \in \cP_f(\La) \sqtx{with}  \abs{\x}\ne \abs{\y}.
\ee
This  justifies the introduction of a modified Hausdorff distance between   finite subsets $\x, \y$ of $\La$:
\be
 \wtilde d_H^\La (\x,\y)= \begin{cases} d_H^\La (\x,\y) & \qtx{if} \abs{\x}=\abs{\y}\\
 \infty & \qtx{otherwise}
 \end{cases}.
\ee

 We consider the following  energy intervals,  labeled by $t\in\R$, and defined by
 \be\label{Im}
I_{\le t}&= \left(-\infty, (t+\tfrac 1 4)\tfd\right),\quad  \mathds{H}_t=  \set{z\in \C:\  \Rea z \in I_{\le t}},\\
 I_{t}&= \left[\tfd, (t+\tfrac 1 4)\tfd\right).
\ee

 We will denote by $\E_\La$ the expectation with respect to the random variables $\set{\omega_i}_{i\in\La}$. 
In this paper  we will use generic constants $C, c$, etc., whose values will be allowed to change from line to line, even  in a multi-line displayed equation. These constants will,  in general, depend on the  fixed parameters of the model  such as $\mu$,  $\Delta$,   $ \lambda$, and on the fractional moment exponent $s$,  but (critically) they will be volume-independent.  We will not indicate the dependence on the fixed parameters and  $s$, but, 
when necessary, we will indicate the dependence of a constant on other parameters, say $q, N, \dots$,   explicitly by writing the constant  as $C_q, C_{q,N}, \ldots$.  If we write $C_q$, it  does not depend on $N$. These constants can always be estimated from the arguments, but we will not track the changes to avoid complicating the arguments.  We  will use $C$ to indicate that the constant should be sufficiently large for a bound to hold, and $c$ to indicate that the constant should be sufficiently small,
but  still requiring $c >0$ . We generally use the same $C$ and $c$ for different constants in the same equation.

 We  are now ready to give the mathematically precise  formulation of Theorem \ref{thm:ppspec}. 
 
  \begin{theorem}\label{thm:ppspec1} 
Fix parameters $\Delta_0>1$ and $ \lambda_0 >0$. Given  $q\in \frac 12 \N^0$, 
 there exists a constant $Y$ (which depends on $\Delta_0\,,\lambda_0$, $\mu$, and $q$) such that, 
 for all $\Delta \ge \Delta_0$ and $\lambda\ge \lambda_0$ satisfying  $\lambda \Delta^2\ge Y$ the following holds:  

\begin{enumerate}

\item  $H^\Z$ exhibits spectral and eigenstate  localization  in the interval  $I_{\le q}$, more precisely, there exists an event $\cE$, with $\P_\Z(\cE)=1$, such that for $\omega \in \cE$  the spectrum of $H^\Z$ in $I_{\le q}$ is pure point, and if $\psi=\psi_\omega$ is an eigenfunction of $H^\Z$ with corresponding eigenvalue  in  $I_{{\le q}}$, so $\cN^\Z\psi =N_\psi \psi$, where   $ N_\psi \in \N^0$,  it decays exponentially   in the following sense:   
\be
\abs{\psi(\y)} \le C_{\omega,N_\psi} \abs{\x_{\psi}}_{2}^{N_\psi +1}\e^{-{c_q} \wtilde d_{H}^\Z({\bf y},\x_{\psi})} \qtx{for all} \y \in \cP_+(\Z), 
\ee
where $\x_{\psi}\in \cP_{N_\psi}(\Z)$ is a center of localization for $\psi$,  that is, it satisfies
\be
\abs{\psi(\x_{\psi})}^2 \ge
  \frac {\pa{\abs{\x_{\psi}}_2+1}^{-\pa{N_\psi+1}}}{\sum_{\u \in  \cP_{N_\psi}(\Z)}       \pa{    \abs{\u}_2+1}^{-\pa{N_\psi +1}}} .
\ee

\item $H^\Z$ exhibits weak dynamical localization in the interval  $I_{{ \le q}}$, more precisely,
\be\label{eq:dloc}
\E_\Z \set{\sup_{f\in B_1(I_{\le q})}\abs{ f(H^\Z) (\x,\y)}}\le  C_q \e^{-c_q \wtilde d_{H}^\Z({\bf x},{\bf y})} \qtx{for all}\x,\y\in \cP_{+}(\Z).
\ee
\end{enumerate}
\end{theorem}

\begin{remark}
 The result above is not as vacuous as  the spectrum $\sigma(H^\Z)=\set{0}\cup[1-\frac1\Delta,\infty)$ with probability one.  (See, e.g., the discussion in  \cite{EKS1}).  Note also that $0$ is a simple eigenvalue.
 \end{remark}

The key input for proving Theorem~\ref{thm:ppspec1} is an immediate corollary to \cite[Theorem 2.4]{EK22}, which we now state.

\begin{theorem}\label{thminput}
Fix the parameters $\Delta_0>1$ and $ \lambda_0 >0$. Let $q\in \frac 12 \N^0$ and $s \in (0,\frac 13)$.
  Then there exists a constant $Y$ (which depends on $\Delta_0\,,\lambda_0$, $\mu$, $q$, and $s$) such that, 
 for all $\Delta \ge \Delta_0$ and $\lambda\ge \lambda_0$ satisfying  $\lambda \Delta^2\ge Y$ the following holds:  
For  all finite $D\subset \Z$  we have
 \be\label{eq:FVC56}
\sup_{z\in\mathds{H}_q} \E_D\set{\abs{G^{D}_{z}({\bf x},{\bf y})}^s}\le C_q \abs{D}^{C_q}\e^{-c_q  \wtilde d^D_H ({\bf x},{\bf y})} \mqtx{for all} {\bf x},{\bf y} \in \cP_+(D).
\ee
\end{theorem}

\begin{proof}
We proved a stronger result  in  \cite[Theorem 2.4]{EK22}, where it is shown that under the hypotheses
of the theorem there exists a constant $Y$ (which depends on $\Delta_0$, $\lambda_0$,  $\mu$, $s$, and $q$) such that, 
 for all $\Delta \ge \Delta_0$ and $\lambda\ge \lambda_0$ satisfying  $\lambda \Delta^2\ge Y$, for all $D\subset \Z$ finite  we have 
\be\label{eq:oldp4589}
\sup_{z\in\mathds{H}_q} \E_D\set{\norm{ P_-^AR^{D}_{z}P_+^B}^s}\le C_q\abs{D}^{C^\pr_q}    \e^{-c_q\dist_D\pa{A,\La \setminus B}},      
\ee
for all  $A\subset B\subset D$ with    $A$ connected in $D$. ({\cite[Theorem 2.4]{EK22}  is stated and proved  for real energies in the  intervals $(-\infty, k+\tfrac 3 4]$, 
 where $k\in \N^0$. The  proof  is also valid for complex energies $z$ with  $ \Rea z  \le (k+\tfrac 3 4)\nfd$, with the same constants.  The above result follows.)}

Given     $\x, \y \in \cP_+ (D)$  with $\abs{\x}=\abs{\y}$, and letting  $r=d_H^D ({\bf x},{\bf y})$,  then either $r=d_D (x, \y)$   for some $x \in \x$, or  $r=d_D (y, \x)$   for some $y \in \y$.  Both cases being similar, we assume the former. In this case, using \eq{eq:Txy} and  \eq{piu}, we have
\be
\abs{G^{D}_{z}({\bf x},{\bf y})}= \norm{\pi_{\x} R_z^D \pi_{\y}}\le  \norm{ \cN_x R_z^D  P_+^{[x]^D_{r-1}}},
\ee
 hence \eqref{eq:FVC56} follows from  \eqref{eq:oldp4589}  as  
   $d_D\pa{\set{x},\La \setminus {[x]^D_{r-1}}}\ge  r$.  
\end{proof}
 We now state our main technical result, Theorem~\ref{cor:weakinf} below. But first we need to introduce some  additional notation  and observations.

 Let $S\subset \Z$.  Given an energy interval $I$, we set $\sigma_I(H^S)=\sigma(H^S)\cap I$.
If  $\nu \in \R$, we set $\pi^S_\nu= \chi_{\set{\nu}} (H^S)$, the spectral projection  of $H^S$ on the set $\set{\nu}$.  
The {\it eigencorrelator}  $\cQ^S_I$ for $H^S$  in the energy interval $I$ is given by 
\be\label{eq:eigcor}
\cQ^S_I({\bf x},{\bf y}) =  \sum_{\nu\in\sigma_I (H^S) }\abs{\pi_{\nu} (\x,\y)}\qtx{for} {\bf x},{\bf y}\in \cP_f(S).
\ee
If $S$ is finite, or, more generally, if $H^S$ has pure point spectrum in $I$, we have 
\be\label{eq:eigcor33}
\cQ^S_I({\bf x},{\bf y})=\sup_{f\in B_1(I)}\abs{f(H^S) ({\bf x},{\bf y})} \qtx{for} {\bf x},{\bf y}\in \cP_f(S).
\ee

We will write   $\sigma_q(H^S)=\sigma_{I_q}(H^S)$  and
$ \cQ^S_q ({\bf u},{\bf v}) =\cQ^S_{I_{\le q}}({\bf u},{\bf v})$ for $q\in \frac 12 \N^0$.

\begin{theorem}[Finite volumes criterion]\label{cor:weakinf} 
Fix $\Delta >1$ and $\lambda >0$.
  Let  $s\in (0,\frac 13)$ and  $q\in \frac 12 \N^0$.    Suppose that for  all finite $D\subset \Z$  we have
 \be\label{eq:FVC}
\sup_{z\in\mathds{H}_q} \E_D\set{\abs{G^{D}_{z}({\bf x},{\bf y})}^s}\le C_q \abs{D}^{C_q}\e^{-c_q  \wtilde d^D_H ({\bf x},{\bf y})} \mqtx{for all} {\bf x},{\bf y} \in \cP_+(D).
\ee
 Then    for all  $\La\subset\Z$     we have
 \be\label{eq:weakinf} 
\sup_{z\in\mathds{H}_q}\E_\La \set{\abs{G^{\La}_{z} ({\bf x},{\bf y})}^s}\le  C_{q}\e^{-c_{q} \wtilde d_H^\La ({\bf x},{\bf y})} \mqtx{for all}  {\bf x},{\bf y} \in \cP_+(\La).
\ee 
      Furthermore, for all $D\subset \Z$ finite we have
 \be \label{eq:smva+}
\E_{D}\set{\cQ_q^D ({\bf x},{\bf y})}\le  C_q e^{-c_q \wtilde d_{H}^D({\bf x},{\bf y})}  \mqtx{for all}\x,\y\in \cP_{+}(D).
\ee    
\end{theorem}
We only consider ${\bf x},{\bf y} \in \cP_+(\La)$  because $G_z^\La(\emptyset,\emptyset)=- \tfrac 1z$ for $z\ne 0$ and  $G_z^\La(\emptyset,\x)=0$ for ${\bf x} \in \cP_+(\La)$.  More generally,    given a  bounded Borel measurable function $f$, we have $f(H^\La)(\emptyset,\emptyset)=f(0)$  and  $f(H^\La)(\emptyset,\x)=0$  for ${\bf x} \in \cP_+(\La)$.

\begin{remark}
The input in the theorem,  the estimate \eq{eq:FVC} (the  finite volumes criterion), allows for volume dependence, whereas  the output \eq{eq:weakinf}   
 completely suppresses this dependence. From a technical point of view, this is one of the delicate points in  the analysis,  and the suppression of  volume dependence is a crucial step in proving  Theorem \ref{cor:weakinf}.   Note that the output  \eq{eq:weakinf}   is also valid  for infinite subsets $\La$ of $ \Z$.
\end{remark}
 The proof of Theorem \ref{cor:weakinf}, presented in Section \ref{sec:proofs}, proceeds by  induction over $q\in \frac 12\N^0$, with constants $C_q$ and $c_q$ in  \eq{eq:weakinf}  that deteriorate with $q$, rendering the method unpractical beyond fixed energy intervals. Similarly to the case of random  Schr\"odinger operators in dimension higher than one, it is not clear whether this restriction is a technical shortcoming or a feature (i.e., there is a phase transition for high energies for this model). There is no consensus among physicists on whether such phase transition occurs or not in the infinite volume system.

Theorem~\ref{thm:ppspec1} follows immediately from Theorem~ \ref{thminput},  Theorem~\ref{cor:weakinf}, and Theorem~\ref{thmdynloc} below.

\begin{theorem}\label{thmdynloc}
Let $q\in \frac 12 \N^0$, and suppose that for all $D\subset \Z$ finite  we have
 \be 
\E_{D}\set{\cQ_q^D ({\bf x},{\bf y})}\le  C_q \e^{-c_q \wtilde d_{H}^D({\bf x},{\bf y})}  \qtx{for all}\x,\y\in \cP_{+}(D).
\ee  
Then 
\be
\E_\Z \set{\sup_{f\in B_1(I_{\le q})}\abs{f(H^\Z) ({\bf x},{\bf y})} }\le  C_q \e^{-c_q \wtilde d_{H}^\Z({\bf x},{\bf y})} \qtx{for all}\x,\y\in \cP_{+}(\Z).
\ee
Moreover, there exists an event $\cE$, with $\P_\Z(\cE)=1$, such that for $\omega \in \cE$  the spectrum of $H^\Z$ in $I_{\le q}$ is pure point, and if $\psi_\omega$ is an eigenfunction of $H^\Z$ with the corresponding eigenvalue  in  $I_{q}$, so $\psi \in \cH_\Z^{N_\psi}$ for some $N_\psi \in \N$,  it decays exponentially   in the following sense:   
\be
\abs{\psi(\y)} \le C_{\omega,N_\psi}   \abs{\x_{\psi}}_2^{N_\psi +1}\e^{-\frac {c_q}2 \wtilde d_{H}^\Z({\bf y},\x_{\psi})}, 
\ee
where $\x_{\psi}\in \cH_\Z\up{N_\psi}$ is a center of localization for $\psi$, that is,
\be
\abs{\psi(\x_{\psi})}^2 \ge
  \frac {\pa{\abs{\x_{\psi}}_2+1}^{-\pa{N_\psi+1}}}{\sum_{\u \in \cH_\Z^{(\N_\psi)}     }          \pa{    \abs{\u}_2+1}^{-\pa{N_\psi +1}}}.
\ee
\end{theorem}
Theorem~\ref{thmdynloc} is proven in Section \ref{sec:main}, where it is derived from \cite[Theorem 4.1]{AW2}.

The rest of this paper is organized as follows: In Section \ref{sec:feat} we introduce notation and collect some basic properties of the XXZ spin chain. 
We prove Theorems \ref{cor:weakinf} and \ref{thmdynloc}  in Sections \ref{sec:proofs} and \ref{sec:main}, respectively. In Appendix \ref{sec:nom} we provide a more detailed discussion of localization types for single-particle and many-body systems.  
Appendix \ref{sec:expsum} provides bounds for the exponential sums  encountered throughout the paper. In Appendix \ref{app:quasil} we discuss key properties of the  filter function that  appears in the proof of  Theorem \ref{cor:weakinf}.

\section{Basic features of the XXZ spin chain}\label{sec:feat}

Let $\La \subset \Z$. When working  with a fixed  $\La $, we write $K^c=\La\setminus K$ for $K\subset \La$. 
If $\x \subset \La$ and $S\subset \La$,  we write $\x_S=  \x \cap S$.  If $P$ is an orthogonal projection, we write $\bar P= \ident  - P$.

 Recall that $\mathcal{N}_i$  is   the orthogonal  projection onto the spin-down state (called the {\it local number operator}) at the site $i$. Given $S\subset \Lambda$, 
 $\cN^{S} = \sum_{i\in S} \mathcal{N}_i$ is the {\it total (spin-down) number  operator} in $S$.
  The total number operator $\cN^\La$  has eigenvalues $0,1,2,\ldots, \abs{\La}$.  We set  $\cH_\Lambda\up{N}=\Ran {\chi_{\set{N}}(\mathcal N^\Lambda)}$,  obtaining   the Hilbert space decomposition 
$ \cH_\La= \bigoplus_{N=0}^{\abs{\La}} \cH_\La\up{N}$.  

The operator $H^\La$ is the sum of three operators,
\be\label{bD}
&H^\La= -\tfrac 1 {2\Delta} \bD^\Lambda +\cW^\Lambda +\lambda   V^\La_\omega, \qtx{where}\\
& \bD^\Lambda = \sum_{\set{i,i+1}\subset \Lambda} \pa{\sigma_i^+\sigma_{i+1}^-+\sigma_i^-\sigma_{i+1}^+},\quad 
 \cW^\Lambda  = \cN^\Lambda -\sum_{\set{i,i+1}\subset \Lambda} \cN_i\cN_{i+1} , \quad
 V^\La_\omega  = \ \sum_{i\in \Lambda} \omega_i \mathcal{N}_i.
 \ee
 
Note that the operators $\set{\cN_i}$ (and thus $\cW^\Lambda$ and $ V_\omega$) are diagonal in the canonical basis:   $\cN_i \phi_{\bf x}=\phi_{\bf x}$ if $i\in {\bf x}$ and $0$ otherwise.  $\cW^\La$ is the \emph {number of clusters operator}:     $\cW^\La \phi_\x=W^\La_{\x}\phi_\x $ for $\x\subset \La$ finite, where $W^\La_{\x}$ is the number of clusters (connected components) of $\x$ as a subset of $\La$, so 
 $\sigma\pa {\cW^\La}\subset \set { 0,1,2,\ldots, \abs{ \La}}$.    $ V^\La_\omega $ is the \emph{random potential}:
 \be\label{bD2}
 V^\La_\omega \phi_\x= V_\omega(\x) \phi_\x \mqtx{for} \x\in \cP_f(\La ), \qtx{where}   V_\omega(\x)=\pa{\sum_{i\in \x} \omega_i}.
 \ee

 An important feature of the XXZ Hamiltonian  $H^\Lambda$  is the conservation of the total particle number (or magnetization):  the operators $ \bD^\Lambda$ and the total  number of particles operator $\cN^\La= \sum_{i\in \La} \cN_i$ commute  (i.e., all bounded functions of these operators commute), and hence $H^\La$ and $\cN^\La$ also commute.  If $\La$ is finite,  this is equivalent to  
\be
 \ [H^\La, \cN^\La]= -\tfrac{1}{2\Delta} [ \bD^\Lambda, \cN^\Lambda]  =0. 
\ee

 It can be verified (e.g.,  \cite{EK22})  that   
\be\label{cWbD}
- 2\cW^\La  \le - \bD^\La \le 2\cW^\La  .
\ee
 Since $\lambda \ge 0$, and  $V_\omega\ge0$ by our assumption on the random variables, it follows that
\beq\label{HtfdW}
H^\La \ge \tfd \cW^\La,
\eeq
and, as a consequence,  the spectrum of $H^\Lambda$ is  of the form  
\be\label{eq:spH}
\sigma(H^\Lambda)=\set{0} \cup \pa{\left[1 -\tfrac 1 \Delta, \infty \right ) \cap  \sigma(H^\Lambda) }.
\ee

Given $N\in \N$, we identify  $\x \in \cP_N(\La)$  with $(x_1,\ldots,x_N)\in \La^N$, where $x_1<\ldots<x_N$.
We introduce a distance   in $\cP_N(\La)$ by
\be\label{eq:l1di}
d_1^\La(\x,\y)=  \sum_{i=1}^N  \dist_\La (x_i,y_i)  \qtx{for}\x,\y \in \cP_N(\La).
\ee
We may have $d_1^\La(\x,\y)=\infty$ if   $\La$ is disconnected.    
Note that 
\be
d_1^\La(\x,\y)& \ge d_1(\x,\y):=d_1^\Z(\x,\y)=\abs{{\bf x}-{\bf y}}_1=\sum_{i=1}^N\abs{x_i-y_i},\\
d_1^\La(\x,\y)& \ge d_H^\La (\x,\y).
\ee

In view of \eq{fxy0}, we also introduce a    distance  on $\cP_f(\La)$ by
\be\label{d1tilda}
 \wtilde d_1^\La(\x,\y)= \begin{cases} d_1^\La(\x,\y) & \qtx{if} \abs{\x}=\abs{\y}\\
 \infty & \qtx{otherwise}
 \end{cases}.
\ee

Since our arguments rely on a certain decoupling idea, we need to introduce graphs obtained from $\La$ by decoupling.  Given $K\subset \La$ (we allow $K=\emptyset$), we consider the graph $\La\up{K}$ with vertex set  $\La$
and edges totally contained in either $K$ or $K^c=\La \setminus K$, so $K$ is disconnected from $K^c$. (In particular, $d_{\La\up{K}} (K,K^c)=\infty$.)
We define  $d_1^{K,K^c }(\x,\y)=  d_1^{\La\up{K} }(\x,\y)$  and  $ \wtilde d_1^{K,K^c }(\x,\y)= \wtilde d_1^{\La\up{K} }(\x,\y)$   as in \eq{eq:l1di} and \eq{d1tilda} using the distance $d_{\La\up{K}} (\cdot,\cdot)$.
Note that  $\wtilde d_1^{\La\up{K} }(\x,\y)\ge \wtilde  d_1^{\La }(\x,\y)$ for all $\x,\y \in \cP_f(\La)$.

We    consider the operators  $H^K= H^K\otimes \mathds{1}_{\cH_{K^c}}$  and   $H^{K^c}=\mathds{1}_{\cH_{K}}\otimes H^{K^c}$acting on $\cH_\La$. Then the decoupled Hamiltonian,  resolvent, and boundary operator on $\cH_\La$ are given by
\be\label{eq:Gamma}
H^{\La\up{K}}=H^{K,K^c}= H^{K}+H^{K^c}, \; R_z^{\La\up{K}}= R_z^{K,K^c}=\pa{H^{K,K^c}-z}^{-1}, \;\Gamma^K=H^\La-H^{K,K^c}.
 \ee
The corresponding decoupled Green function is then 
\be\label{eq:Greend}
G^{\La\up{K}}_{z}({\bf x},{\bf y})= G^{K,K^c}_{z}({\bf x},{\bf y}) = \langle \phi_{\bf x},R_{z}^{K,K^c} \phi_{\bf y}\rangle.
\ee
Note that $G^{K,K^c}_{z}({\bf x},{\bf y})$ vanishes unless $\abs{{\bf x}\cap K}=\abs{{\bf y}\cap K}$ and $\abs{{\bf x}\cap K^c}=\abs{{\bf y}\cap K^c}$.

Since 
on $\Ran \chi_{\N}(\cN^\La)$ we have
\be\label{HWE12}
H^{\La\up{K}}-E \ge \tfd \cW^{\La\up{K}}- \tfrac 34 \tfd \ge \tfrac 14 \tfd \cW^{\La\up{K}}  \sqtx{for} E \in I_{\le \frac12},
\ee
 we have, using  the resolvent identity,  that
 \be\label{WRW}
  \sup_{z\in \mathds{H}_{\frac 12}}  \norm{\pa{\cW^{\La\up{K}}} ^{\frac 12} R_z^{\La\up{K}}\pa{\cW^{\La\up{K}} }^{\frac 12}} \le 2 \sup_{E\in I_{\le \frac12}} \norm{\pa{\cW^{\La\up{K}}} ^{\frac 12} R_E^{\La\up{K}}\pa{\cW^{\La\up{K}} }^{\frac 12}}\le   \tfrac 8 {\fd} <\infty.
 \ee
 
The estimate \eq{WRW}  is the required hypothesis for 
an important property of Green functions,  a Combes-Thomas bound: 
 \be\label{eq:CT}
 \sup_{z\in \mathds{H}_{m}} \abs{G^{\La\up{K}}_z({\bf x},{\bf y})}\le C_{m} \e^{-c_{m} \wtilde d_1^{\La\up{K}}(\x,\y)} 
  \qtx{for}  m  \le \tfrac 12 \qtx{and}
\x ,\y \in \cP_+(\La),
\ee
which holds for all $K\subset \La$.

 The bound \eq{eq:CT} is proven similarly to   \cite[Proposition 4.1]{EKS1}, which  is stated for   $H^{\La\up{K}}_N$ with $N\in \N$, and yields decay in the distance 
$d_1^\Z(\x,\y)=\abs{\x-\y}_1$.  To get the bound  \eq{eq:CT}, in the proof we replace the distance $d_1^\Z(\x,\y)$ by $d_{1,R}^{\La\up{K}}(\x,\y)=\min \set{d_1^{\La\up{K}}(\x,\y),R}$, where $R\gg 1$,  perform the proof, and let $R\to \infty$. We get the bound \eq{eq:CT} for $H^{\La\up{K}}_N$ with constants independent of $N$, $\La$ and $K$, so combining with  the definition \eq{d1tilda} we get \eq{eq:CT}.

\begin{remark}
As mentioned in Section~\ref{secmodel}, $H^\La_N$ (or $H^{\La\up{K}}_N$)  is a random Schr\"odinger operator on a certain subgraph of $\La^N$, but the standard	 Combes-Thomas bounds will only yield \eq{eq:CT} with $N$-dependent constants.
Although \cite[Proposition 4.1]{EKS1} is stated and proved on each $\cH_\La\up{N}$, the proof uses the special structure of  $H^\La_N$ (or $H^{\La\up{K}}_N$)  to obtain constants independent of $N$.  In particular, the required hypothesis is    $\norm{\pa{\cW^{\La\up{K}}} ^{\frac 12} R_z^{\La\up{K}}\pa{\cW^{\La\up{K}} }^{\frac 12}} <\infty$, not just   $\norm{ R_z^{\La\up{K}}} <\infty$ as for the standard	 Combes-Thomas bounds. \end{remark}
Unfortunately,  we cannot use the Combes-Thomas bound given in  \eq{eq:CT} directly for energies lying in $\mathds{H}_m$ for  $ m \ge  1$.  The way around this obstacle 
is to lift the spectrum of the operator $H^{\La\up{K}}$, and show that \eq{eq:CT}  holds for the lifted operator.  

To simplify notation, for the remainder of this section we will simply write $H^\La$ for $H^{\La\up{K}}$,   with the results holding for $H^{\La\up{K}}$ for all $K\subset \La$.

Given  $m \in \N^0$, we set $Q_m^\Lambda=\chi_{\set{m}}\pa{\cW^\Lambda}$, the orthogonal projection onto  configurations ${\bf x}$ with exactly $m$ clusters, and let  $Q_B^\Lambda=\chi_{B}\pa{\cW^\Lambda}=\sum_{m\in B} Q_m^\Lambda$  for $ B\subset \N^0$.
 For $k\in \N$, we set 
\be\label{QkhatQ}
Q_{\le k}^\Lambda   =Q_{\set{1,2,\ldots,k}}^\Lambda =\sum_{ m=1}^k Q_m^\Lambda \qtx{and} 
\what Q_{\le k}^\Lambda   =Q_{\le k}^\Lambda + \tfrac {k+1} k Q_0^\Lambda.\,
\ee
 and recall   that \cite[Lemma~3.5]{EK22}
\begin{align}\label{trXk}
\norm{Q_{\le k}^{\Lambda}}_{HS}&\le \sqrt{k} \abs{\Lambda}^{k},\\
\label{trkH}
\tr \chi_{\what I_{\le k}}(H^\La)&\le  k\abs{\Lambda}^{2k}+1.
\end{align}
 Given  $q\in\frac12\N^0$, we set
\be \label{eq:compH'}
\what  H_q^{ \Lambda}=\begin{cases}H^{\Lambda}+\tfd Q_0^\Lambda & \qtx{if} q=0,\tfrac12\\
H^{\Lambda}+{\cl{q}}\tfd \what  Q_{\le \cl{q}}^{\Lambda} & \quad \text{otherwise} \end{cases} .
\ee 
Given  $ z\notin \sigma(\what  H_q^\La) $, we set 
\be\label{eq:Greenk}
\what  R^{\Lambda}_{q,z}&= \pa{\what  H_q^\La  -z}^{-1} \qtx{and}\what G^\La_{q,z}({\bf x},{\bf y}) = \langle \phi_{\bf x},\what R_{q,z}^\Lambda \phi_{\bf y}\rangle.
\ee 
 These operators satisfy
\be \label{eq:hatH1'}
& \what  H_q^{ \Lambda} \ge \tfd\sqtx{for}   q=0,\tfrac12 ,\qtx{and}  \what  H_q^{ \Lambda}\ge  \pa{{\cl{q}}+1} \tfd  \quad \text{otherwise},
 \ee   
and, moreover, on $\Ran \chi_{\N}(\cN^\La)$ we have
\be\label{HWE1244}
\what  H_q^{ \Lambda}-E \ge  \tfrac 1 {4\pa{{\cl{q}}+1} }\tfd \cW^\La \qtx{for all}    E \in I_{\le {q}} \sqtx{for} q\in \tfrac 12 \N^0,
\ee
so it follows that
 \be\label{WRW44}
  \sup_{z\in \mathds{H}_m}  \norm{\pa{\cW^{\La}} ^{\frac 12} \what  R^{\Lambda}_{q,z}\pa{\cW^{\La} }^{\frac 12}} \le 2 \sup_{E\in I_{\le \frac12}} \norm{\pa{\cW^{\La}} ^{\frac 12} \what  R^{\Lambda}_{q,E}\pa{\cW^{\La} }^{\frac 12}}\le  \tfrac  {8\pa{{\cl{q}}+1} }{\fd }<\infty.
 \ee

The bound \eq{HWE1244} is just \eq{HWE12} for $q=0,\frac12$. To derive  \eq{HWE1244} for $q\ge 1$   on $\Ran \chi_{\N}(\cN^\La)$, we use
\be
\what  H_q^{ \Lambda}-E\ &\ge \tfd \cW^\La +{\cl{q}}\tfd \what  Q_{\le \cl{q}}^{\Lambda}  -E \\
&\ge
\tfd\pa{  Q_{\le \cl{q}}^{\Lambda}\pa{\cW^\La + {\cl{q}}- (q+\tfrac 14)}   +
\bar Q_{\le \cl{q}}^{\Lambda}  \cW^\La\pa{ 1 -  \tfrac {q+\frac 14}{\cl{q}+1}     }} \\
& \ge \tfrac 1 {4\pa{{\cl{q}}+1} }\tfd \cW^\La .
\ee

It follows from \eq{WRW44} that we have a 
 Combes-Thomas bound  similar to  \eq{eq:CT}  for the modified Green functions:
For all $q \in \frac 12 \N^0$ we have
 \be\label{eq:CTk}
 \sup_{z\in \mathds{H}_q} \abs{\what G^\La_{q,z}({\bf x},{\bf y})}\le C_q \e^{-c_q \wtilde d_1^\La(\x,\y)} 
 \qtx{for all}  {\x,\y \in \cP_+(\La).}
\ee

 We observe that, given $S\subset \Z$,   $q\in\frac12\N^0$, and $\nu \in \sigma_q(H^S)$, it follows from \eqref{eq:compH'} that
\be\label{eigest}
\pi^S_\nu=\lceil q \rceil\tfd \what  R^{S}_{q,\nu}\what Q^S_{\le \lceil q \rceil} \pi^S_\nu=
\lceil q \rceil^2\tfd^2 \what  R^{S}_{q,\nu}\what Q^S_{\le \lceil q \rceil} \pi^S_\nu  \what  R^{S}_{q,\nu}\what Q^S_{\le \lceil q \rceil}.
\ee

The proof of Theorem \ref{cor:weakinf}  will be facilitated by
the following    a-priori estimate (see, e.g., \cite[Lemma 3.4]{EK22}):
 \be\label{eq:weak1-1a} 
\E_{\set{i,j}}\pa{\norm{T_1\cN_iR_z^\Lambda\cN_j T_2}_{2}^{s^\pr}}  \le
{C_{s^\pr}} \lambda^{-{s^\pr}}\norm{T_1}^{s^\pr}_{2}\norm{T_2}^{s^\pr}_{2} \sqtx{for all} z\in \C \sqtx{and} {s^\pr}\in (0,1) ,
\ee 
where  $\norm{ \cdot}_2$ denotes the Hilbert-Schmidt norm, which
implies that 
\be \label{eq:weak1-1} 
 \sup_{z\in\C}\ \E_{ \La}\set{\abs{G^{ \La}_{z}({\bf x},{\bf y})}^{s^\pr}}\le C_{s^\pr} \mqtx{for all}  {\bf x},{\bf y}\in \cP_+(\La)  \mqtx{and} {s^\pr}\in (0,1).
 \ee

\section{Proof of Theorem \ref{cor:weakinf}}
\label{sec:proofs}

In this Section we prove Theorem \ref{cor:weakinf}.   We fix  $\Delta >1$, $\lambda >0$, and  $s\in (0,\frac 13)$.
 For  $k \in \set{1,2,\ldots,N}$ we let
 \be
 \cP_{N,k}(\La)=\set{\x\in \cP_N(\La), 1\le W_\x^\La  \le k}=  \set{\x\in \cP_N(\La), \phi_\x \in \Ran Q^\Lambda_{\le k}}.
 \ee

\begin{lemma}\label{lemweakstrong}
Let $q\ \in  \frac 12 \N$, $1\le q $, and  $N\in \N$.  Fix $\La \subset \Z$, and and suppose  \eq{eq:weakinf} holds for all  ${\bf x},{\bf y}\in \cP_{N,\cl{q}}(\La)$.  Then \eq{eq:weakinf} holds for all ${\bf x},{\bf y}\in \cP_{N}(\La)$ (with different constants, independent of $\La$ and $N$).
\end{lemma}

\begin{proof}

We use the following resolvent identity:
 \be\label{eq:resmodl}
R_{z}^{\Lambda}=\what  R^{\Lambda}_{q,z}+ \clq\tfd R^{\Lambda}_{z} \what Q^\Lambda_{\le \clq}\what  R^{\Lambda}_{q,z}=\what  R^{\Lambda}_{q,z}+ \clq \tfd\what  R^{\Lambda}_{q,z}\what Q^\Lambda_{\le \clq}  R^{\Lambda}_{z}.
\ee
Applying it twice, we get 
 \be\label{eq:resmod'1}
R_{z}^{\Lambda}=\what  R^{\Lambda}_{q,z}+ \clq\tfd\what  R^{\Lambda}_{q,z} \what Q^\Lambda_{\le \clq} \what  R^{\Lambda}_{q,z}+\clq^2\tfd^2\what  R^{\Lambda}_{q,z}\what Q^\Lambda_{\le \clq}   R^{\Lambda}_{z} \what Q^\Lambda_{\le \clq}\what  R^{\Lambda}_{q,z}.
\ee

Suppose now  that \eqref{eq:weakinf} holds for all  ${\bf u},{\bf v}\in \cP_{N,\cl{q}}(\La)$.   Then,   using  also \eqref{eq:resmod'1} and   \eqref{eq:CTk}, we can bound
\be\label{eq:genco}
 \sup_{z\in\mathds{H}_q}\sup_{\La\subset \Z}\E_\La\set{\abs{G^{ \La}_{z}({\bf x},{\bf y})}^s}&
 \le C_q  \e^{-c_qd_1^\La(\x,\y)}+C_q\sum_{{\bf u}\in \cP_{N,\cl{q}}(\La)}\e^{-c_qd_1^\La(\x,\u)}
 \e^{-c_q\d_1^\La(\u,\y)}\\& \quad  \quad +C_q  \sum_{{\bf u},{\bf v}\in \cP_{N,\cl{q}}(\La)}\e^{-c_qd_1^\La(\x,\u)}\e^{-c_qd^\La_H\pa{{\bf u},{\bf v}}}\e^{-c_qd_1^\La(\v,\y)}\\ &\le  C_q    e^{-c_qd^\La_H\pa{{\bf x},{\bf y}}},
\ee   
where in the last step we used  properties of exponential sums  (see \eqref{eq:d1bn43}).
\end{proof}

\begin{proof}[Proof of  Theorem \ref{cor:weakinf}]

  We take  $q\in \frac 12 \N^0$, and assume     that \eq{eq:FVC} holds for  all finite $D\subset \Z$.   Given $\La \subset \Z$,  in view of \eq{fxy0} we only have to prove  \eq{eq:weakinf}   for   ${\bf x},{\bf y}\in \cP_+ (\La)$  with  $\abs{\x}=\abs{\y}$.
  
  The proof will proceed by induction on  $q\in \frac 12 \N^0$.   For $q=0,\frac 12$, the theorem  (i.e.,  \eq{eq:weakinf}) follows from the Combes Thomas bound \eq{eq:CT}.  Given $q\in \frac 12 \N^0$, $q\ge 1$, we l  assume the theorem holds  for  $q-\frac 12$, 
 and will  prove it then holds for $q$.
  
 The proof proceeds by a series of Lemmas. In view of Lemma~\ref{lemweakstrong}, 
 it suffices to prove  \eqref{eq:weakinf}  for all  ${\bf x},{\bf y}\in \cP_{N,\cl{q}}(\La)$.

\begin{lemma}
\label{thm:eigencorweak} 
Let $D\subset \Z$ be finite, let $N\in \N$,  and assume
\be\label{eq:weakinf46} 
 \sup_{z\in\mathds{H}_{ q}}\E_D \set{\abs{G^{D}_{z}({\bf x},{\bf y})}^s}\le  C_q \e^{-c_{q}   d_H^D ({\bf x},{\bf y})} \mqtx{for all}\x,\y\in \cP_{N,\cl{q}}(D).
\ee 
Then      
  \be \label{eq:smva'}
\E_{D}\set{\cQ_q^D ({\bf x},{\bf y})}\le  C_q e^{-c_q d_{H}^D({\bf x},{\bf y})} \ \mqtx{for all}\x,\y\in \cP_{N}(D).
\ee  
 \end{lemma}

\begin{proof}  Let $D\subset \Z$ finite, $N\in \N$, and $\x,\y\in \cP_{N}$.
We  assume that there is $x\in {\bf x}$ such that 
 \be\label{defHD}
  \dist_D (x, {\bf y})= d^D_{H}({\bf x}, {\bf y}),
 \ee 
  the other case being similar. 

 We first prove the lemma for $\x,\y\in \cP_{N,\cl{q}}(D)$.
This is done using the reduction to  resolvents  achieved by using   the estimate \cite[Eq. (7.44)]{AW} and the spectral averaging as in {\cite[Theorem 4.5]{AW2}}.  The final result can be re-formulated in our setting as:  

\emph{Let $r\in(0, 1)$, $N\in \N$, and  let $I\subset \R$ be an interval.   Then for all finite $D \subset \Z$ and 
  ${\bf x},{\bf y}\in \cP_N(D)$  we  have 
 \be\label{eq:eigencorwea} 
\E_{D}\set{\cQ_I^D({\bf x},{\bf y})}\le C_r\sum_{{\bf u}\in \cP_N(D): x\in {\bf u}} \int_I \E_{D}\set{\abs{G_E^D ({\bf u},{\bf y})}^r}dE
 \qtx{for any} x\in {\bf x}. 
 \ee}
Note that that  $d^D_{H}({\bf u}, {\bf y})\ge d^D_{H}({x}, {\bf y})$  if $x\in \u \subset D$.

 Given $\x,\y\in \cP_{N,\cl{q}}(D)$,  we estimate $\E_{D}\set{\cQ_q^D({\bf x},{\bf y})}$ by  \eqref{eq:eigencorwea}, and estimate the term 
$\E_{D}\set{\abs{G_E^D ({\bf u},{\bf y})}^s}$ inside the integral as in \eqref{eq:genco}, using \eq{eq:weakinf46} ,  getting
 \be \label{eq:smva}
\E_{D}\set{\cQ_q^D({\bf x},{\bf y})}&\le C_q \sum_{{\bf u}\in \cP_N(D): x\in {\bf u}}\e^{-c_q d_1^\D(\u,\y)}+C_q  \sum_{{\bf u}\in \cP_N(D): x\in {\bf u}}\ \sum_{{\bf v}\in \cP_{N,\cl{q}}(D)}\e^{-c_qd_1^\D(\u,\v)}\e^{-c_q\d_1^\D(\v,\y)}\\& \quad  \quad +C_q\sum_{{\bf u}\in \cP_N(D): x\in {\bf u}}\ \sum_{{\bf v},{\bf w}\in \cP_{N,\cl{q}}(D)}\e^{-c_qd_1^\D(\u,\v)}\e^{-c_qd^D_H\pa{{\bf v},{\bf w}}}\e^{-c_q\d_1^\D({\bf w},\y)}.
\ee
To bound the first sum, we note that it follows from \eqref{defHD} that
\be\label{defHD32}
d_1^\D(\u,\y)\ge d_H^D({{\bf u},{\bf y}})\ge  \dist_D (x, {\bf y})=d_{H}^D({\bf x}, {\bf y}) \sqtx{if} x \in \u.
\ee
Hence 
\be\label{eq:smva2}
\sum_{{\bf u}\in \cP_N(D): x\in {\bf u}}\e^{-c_q d_1^\D(\u,\y)} \le \e^{-\frac {c_q}2d_{H}^D({\bf x}, {\bf y})}\sum_{{\bf u}\in \cP_N(D)} \e^{-\frac {c_q}2\d_1^\D(\u,\y)}   \le   C_{q} \e^{-\frac  {c_q} 2d_{H}^D({\bf x}, {\bf y})},
\ee  
 where in the last step we used \eqref{eq:d1bn}  and $\y\in \cP_{N,\cl{q}}(D)$.
 
 To estimate the second sum in \eqref{eq:smva}, we use  the triangle inequality to conclude that 
\be
d_1^\D(\u,\v)+d_1^\D(\v,\y)&\ge d_1^\D(\u,\y)\ge d_{H}^D({\bf x},{\bf y}) \qtx{if} x \in \u,\\
d_1^\D(\u,\v)+d_1^\D(\v,\y)&\ge \tfrac12\pa{d_1^\D(\u,\y)+\d_1^\D(\v,\y)}.
\ee  
Hence 
\be\label{eq:smva3}
& \sum_{{\bf u}\in \cP_N(D): x\in {\bf u}}\ \sum_{{\bf v}\in \cP_{N,\cl{q}}(D)}\e^{-c_q d_1^\D(\u,\v)}\e^{-c_qd_1^\D(\v,\y)}\\
 & \quad \le \e^{-\frac {c_q}2d_{H}^D(({\bf x},{\bf y})} \sum_{{\bf u}\in \cP_N(D)}\sum_{{\bf v}\in \cP_{N,\cl{q}}(D)}\e^{-\frac {c_q}4d_1^\D(\u,\y)}e^{-\frac {c_q}2d_1^\D(\v,\y)}
 \le  C_qe^{-\frac {c_q}2 d_{H}^D({\bf x}, {\bf y})},
\ee
using   $\y\in \cP_{N,\cl{q}}(D)$  and \eqref{eq:d1bn}  twice in the last step.

Finally, to estimate the last sum in \eqref{eq:smva}, we use  the triangle inequality  and \eq{defHD32} to conclude that  
\be
&d_1^D(\u,\v)+d_H^D\pa{{\bf v},{\bf w}}+d_1^D({\bf w},{\bf y})\ge d_{H}^D({\bf x},{\bf y}),\\
& d_1^D(\u,\v)+d_H^D\pa{{\bf v},{\bf w}}+d_1^D({\bf w},{\bf y})\ge  d_1^D(\u,\v)+ \frac12\pa{d_H^D\pa{{\bf v},{\bf y}}+d_1^D({\bf w},{\bf y})}.
\ee 
Hence 
\be\label{eq:smva4}
& \sum_{{\bf u}\in \cP_N(D): x\in {\bf u}}\ \sum_{{\bf v},{\bf w}\in \cP_{N,\cl{q}}(D)}\e^{-c_q d_1^D(\u,\v)}\e^{-c_qd^D_H\pa{{\bf v},{\bf w}}}\e^{-c_qd_1^D({\bf w},{\bf y})}\\ 
& \qquad \le e^{-\frac  {c_q}2d_{H}^D({\bf x}, {\bf y})}\sum_{{\bf u}\in \cP_N(D):}\ \sum_{{\bf v},{\bf w}\in \cP_{N,\cl{q}}(D)}e^{-\frac  {c_q}2d_1^D(\u,\v)}e^{-\frac  {c_q}4d_H^D\pa{{\bf v},{\bf w}}}e^{-\frac {c_q}4d_1^D({\bf w},{\bf y})}\\
& \qquad \le C_q N^{2\cl{q}}e^{-\frac  {c_q}2d_{H}^D({\bf x}, {\bf y})},
\ee  
using $\y\in \cP_{N,\cl{q}}(D)$,   \eqref{eq:dhbn} and \eqref{eq:d1bn}  in the last step. 

  Putting together \eqref{eq:smva}, \eqref{eq:smva2}, \eqref{eq:smva3}, and  \eqref{eq:smva4} , we get
\be \label{eq:smva5}
\E_{D}\set{\cQ_q^D ({\bf x},{\bf y})}\le  C_q N^{2\cl{q}}\e^{-c_q d_{H}^D({\bf x},{\bf y})}  \mqtx{for all}\x,\y\in \cP_{N,\clq}(D).
\ee

 To remove the $N$ dependence in \eqref{eq:smva5}, we will show
\be\label{eq:almost1}
\E_{D}\set{\cQ_q^D({\bf x},{\bf y})}\le C_q \e^{-c_q N}\mqtx{for all}\x,\y\in \cP_{N,\clq}(D),
\ee
using  a large deviation estimate.
Let $\bar{\mu}= \E \set{\omega_0}$, and assume $N \lambda \bar{\mu} > 2 \clq\tfd$.  Then the standard large deviations estimate (recall \eqref{bD2})
gives
\be
\P\set{  \lambda V_\omega ({\bf u})< \clq \tfd}\le \P\set{V_\omega ({\bf u})< N \tfrac {\bar{\mu}} 2}\le \e^{- c_\mu N}\sqtx{for all} {\bf u}\in \cP_{N,\clq}(\Z).
\ee
Thus,  for any $N\in \N$,   letting $S=[{\bf x}]^D_N$, and defining the event
\be\label{eventENS}
\cE^S_N= \set{ \lambda  V_\omega ({\bf u})< {\clq\tfd} \sqtx{for some} {\bf u} \in  \cP_{N,\cl{q}}(S)},
\ee
 we have 
  \begin{align}\label{PV<}
 \P\set{\cE^S_N} \le  C_{\mu,q}  \abs{\cP_{N,\cl{q}}(S)} \e^{- c_\mu N} \le C_{\mu,q} \clq \pa{N(2N+1)}^{2\clq } \e^{- c_\mu N} \le C_{\mu,q}^\pr \e^{- c_{\mu,q}^\pr N},
  \end{align}
  where we used $\abs{\cP_{N,\cl{q}}(S)}=\tr Q_{\le k}^{S}$, \eq{trXk}, and $\abs{S}\le N (2N+1)$.   Moreover,
  on the  complementary event   $\pa{\cE_N^S}^c $ we have
    \be\label{PVP}
\chi_N(\mathcal N^S)H^S\ge (\clq+1)\tfd\chi_N(\mathcal N^S) ,
  \ee
  so we can use \cite[Proposition 4.1]{EKS1} to obtain  the Combes-Thomas bound 
 \be\label{eq:CTkS}
\sup_{z\in \mathds{H}_{\clq}}\abs{ G^S_{z}({\bf x},{\bf y})}\le C_q e^{-c_q d_1^S({\bf x},{\bf y})} \qtx{for all} \x,\y\in \cP_{N,\cl{q}}(S).
\ee

To show \eqref{eq:almost1}, we observe that for  $\nu \in \sigma_q (H^D)$ we have 
\be
\pi_{_\nu}\phi_{\bf x}=\pi_{\nu}\pa{H^{S,S^c}-\nu}R_\nu^{S,S^c}\phi_{\bf x}=\pi_{\nu}\Gamma^S R_\nu^{S,S^c}\phi_{\bf x},
\ee
and  the construction of $S$ yields
\be
R_\nu^{S,S^c}\phi_{\bf x}=P_+^{S^c}R_\nu^{S}\phi_{\bf x}.
\ee
It follows that on the  complementary event   $\pa{\cE_N^S}^c $ we have 
\be\label{eq:alm}
\sum_{\nu\in\sigma_q(H^D)} \abs{\langle\phi_{\bf y},\pi_{\nu}\phi_{\bf x}\rangle}&\le \sum_{\nu\in\sigma_q(H^D)}\sum_{{\bf u}\in\cP_N^{\partial_{ex}}(D)}\abs{\langle \phi_{\bf y},\pi_{\nu}\phi_{\bf u}\rangle}\abs{\langle \phi_{\bf u},\Gamma^S P_+^{S^c} R_\nu^{S}\phi_{\bf x}\rangle}\\ &\le C_q\sum_{\nu\in\sigma_q(H^D)}\sum_{{\bf u}\in\cP_N^{\partial_{ex}}(D)}e^{-c_q d_1^D({\bf u},{\bf x})}\abs{\langle \phi_{\bf y},\pi_{\nu}\phi_{\bf u}\rangle}, 
\ee
where we used \eq{eq:CTkS}, and  
\be
\cP_N^{\partial_{ex}}(D)=\cP_N(D, \partial_{ex}^D S) = \set{{\bf u}\in\cP_N(D), {\bf u}\cap\partial_{ex}^D S\neq\emptyset}.
\ee

We next observe that by the Cauchy-Schwarz inequality,
\be\label{eq:CSs}
\pa{\sum_{\nu\in\sigma_q (H^D)}\abs{\langle \phi_{\bf y},\pi_{\nu}\phi_{\bf u}\rangle}}^2 &\le \sum_{\nu\in\sigma_q (H^D)}\langle \phi_{\bf y},\pi_{\nu}\phi_{\bf y}\rangle\,\sum_{\nu\in\sigma_q (H^D)}\langle \phi_{\bf u},\pi_{\nu}\phi_{\bf u}\rangle\\ &= \langle \phi_{\bf y},\chi_{I_{q}}(H^D)\phi_{\bf y}\rangle\, \langle \phi_{\bf u},\chi_{I_{q}}(H^D)\phi_{\bf u}\rangle\le 1.
\ee
Plugging this into \eqref{eq:alm}, we see that
\be\label{eq:almo}
\sum_{\nu\in\sigma_q (H^D) }\abs{\langle\phi_{\bf y},\pi_{\nu}\phi_{\bf x}\rangle}\le C_q\sum_{ {{\bf u}\in\cP_N^{\partial_{ex}}(D)}}e^{-c_q d_1^D({\bf u},{\bf x})}.
\ee 
Since for any ${\bf u}\in\cP_N^{\partial_{ex}}(D)$ we have 
$d_1^D({\bf u},{\bf x})\ge N$, 
we deduce from \eqref{eq:almo}  and \eqref{eq:d1bn} (recall $\y \in \cP_{N,\clq}(D)$), that on $ \cE_N^c$ we have 
\be\label{eq:almos}
\cQ_q^D ({\bf x},{\bf y})=\sum_{\nu\in\sigma_q (H^D)}\abs{\langle\phi_{\bf y},\pi_{\nu}\phi_{\bf x}\rangle}\le C_q \e^{-c_qN}.
\ee
Hence
\be\label{eq:almost}
\E_D\set{\cQ_q^D ({\bf x},{\bf y})}&=\E_D\set{\chi_{\cE^S_N}\cQ_q^D ({\bf x},{\bf y})}+\E_D\set{\chi_{\pa{\cE_N^S}^c }\cQ_q^D ({\bf x},{\bf y})}  \\&\le \P\pa{\cE_N^S}+C_q\e^{-c_q N}\le C_qe^{-c_qN},
\ee
where we have used \eq{PV<} and  \eqref{eq:almost}. The relation \eqref{eq:almost1} follows.

Using \eq{eq:smva5} for   $ d_{H}^D({\bf x},{\bf y})   \ge N$ and  \eqref{eq:almost1} for $ d_{H}^D({\bf x},{\bf y})   <N$ we get
\be \label{eq:smva79}
\E_{D}\set{\cQ_q^D ({\bf x},{\bf y})}\le  C_q\e^{-c_q d_{H}^D({\bf x},{\bf y})}  \mqtx{for all}\x,\y\in \cP_{N,\clq}(D).
\ee

 We now consider the general case $\x,\y\in \cP_N(D)$.  For any $\nu \in \sigma_{\clq}  (D)$, using 
\eq{eigest}, we have
\be
\pi_\nu^D (\x,\y)= C_q  \sum_{\u,\v \in \cP_{N,\clq}(D)}    \what G_{\clq,\nu}^D(\x,\u) \pi_\nu^D (\u,\v){\what G}_{\clq,\nu}^D(\v,\y)
\ee
It follows that 
\be
\cQ^D(\x,\y) &\le \sum_{\u,\v \in \cP_{N,\clq}(D)} \abs{ \what G^D_{\clq,\nu}(\x,\u) }\cQ^D (\u,\v)\abs{  \what G^D_{\clq,\nu}(\v,\y)}\\
&  \le C_q \sum_{\u,\v \in \cP_{N,\clq}(D)}  \e^{-c_qd_1^D({\bf x},{\bf u})}\cQ^D (\u,\v) \e^{-c_q d_1^D({\bf v},{\bf y})},
\ee
where we used \eqref{eq:CTk}.  Using \eq{eq:smva79}, we conclude that
\be
\E_D\set{\cQ^D(\x,\y} &\le C_q  N^{2\clq} \sum_{\u,\v \in \cP_{N,\clq}(D)}  \e^{-c_q d_1^D({\bf x},{\bf u})}\e^{- c_q d_H^D(\u,\v) } \e^{-c_qd_1^D({\bf v},{\bf y})}\\
& \le  C_q  \e^{- c^\pr_q d_H^D(\x,\y)} \sum_{\u,\v \in \cP_{N,\clq}(D)}  \e^{-c^\pr_qd_1^D({\bf x},{\bf u})}\e^{-c^\pr_q\d_1^D({\bf v},{\bf y})} \\
& \le  C_q \e^{- c_q d_H^D(\x,\y)} ,
\ee  
where we used  \eq{eq:d1bn43} twice. This concludes the proof.
\end{proof}

\begin{lemma}\label{prop:decbn}   
Let $q\ \in  \frac 12 \N$, $1\le q $, and   assume the the induction hypothesis,  that is,  Theorem~\ref{cor:weakinf} is proven for $q-\frac 12$. 
Let  $\La \subset \Z$ and  $N\in \N$.
 Then for all $\x,\y\in \cP_{N,\cl{q}}(\La)$ and  any finite  connected set $D\subset\La$ satisfying $ (\x\cup \y)_D \neq\emptyset$ and $(\x\cup \y)_{D^c}\neq\emptyset$ we have
\be\label{eq:weakinfdec}
  \sup_{z\in\mathds{H}_q}\E_\La\set{\abs{G^{D,D^c}_{z}({\bf x},{\bf y})}^s}\le C_{ q} \abs{D}^{2\cl{q}} \e^{-c_{q-\frac12} { d_H^\La ({\bf x},{\bf y})} }.
\ee 
\end{lemma}

\begin{proof} 
Let $\x,\y\in \cP_{N,\cl{q}}(\La)$,  and consider a  finite  connected set  $D \subset \La$ satisfying $ (\x\cup \y)_D \neq\emptyset$ and $(\x\cup \y)_{D^c}\neq\emptyset$.   We only  need  to consider the case  $1\le \abs{\x _ D}= \abs{\y_D}\le N-1  $, as otherwise $G^{D,D^c}_{z}({\bf x},{\bf y})=0$.

To do so, given  $a\in\R$,  let $F_{a}$ be the  analytic  function on $\R$ given by
  \be\label{defF}
   F_{\xi,a} (x)&= \frac{1-\e^{-\xi x^2}}{x-ia} \qtx{for} x \in \R \qtx{if} a\ne 0,\\
    F_{\xi,0} (x)&= \frac{1-\e^{-\xi x^2}}{x} \qtx{for} x \in  \R\setminus \set{0} \qtx{and} F_{\xi,0} (0)=0.
   \ee 
Given $z\in\mathds{H}_q$, let $E=\Rea z$ and $a=\Ima z$.   Setting $r= d_H^\La \pa{{\bf x},{\bf y}}-1$, and taking $F_a (x)=  F_{\xi,a} (x)$ with  $ \xi=   \frac {\Delta^2}{50}  r  $, we have the following bound (recall \eqref{eq:Txy}):  
\be\label{eq:decen}
& \abs{G^{D,D^c}_{z}({\bf x},{\bf y})}\le \abs{\pa{R^{D,D^c}_{z}P_{I_{\le q-\frac12}}(H^{D})}({\bf x},{\bf y})}+\abs{\pa{R^{D,D^c}_{z}\bar P_{I_{\le q-\frac12}}(H^{D})}({\bf x},{\bf y})}\\
&  \quad \le \abs{\pa{R^{D,D^c}_{z}P_{I_{\le q-\frac12}}(H^{D})}({\bf x},{\bf y})}+   \abs{\pa{F_{a}(H^{D,D^c}-E)\bar P_{I_{\le q-\frac12}}(H^{D})}({\bf x},{\bf y})}\\ &\hspace{2cm}  + \abs{\pa{R_z^{D,D^c}\exp\pa{-r(H^{D,D^c}-E)^2}P_{I_{\le q+\frac12}}(H^{D,D^c})\bar P_{I_{\le q-\frac12}}(H^{D})}({\bf x},{\bf y})}\\ &\hspace{2cm} + \abs{\pa{R_z^{D,D^c}\exp\pa{-r(H^{D,D^c}-E)^2}\bar P_{I_{\le q+\frac12}}(H^{D,D^c})\bar P_{I_{\le q-\frac12}}(H^{D})}({\bf x},{\bf y})}.
\ee

For  $z\in\mathds{H}_q$ we have 
\be
&\abs{\pa{R_z^{D,D^c}\exp\pa{-r(H^{D,D^c}-E)^2}\bar P_{I_{\le q+\frac12}}(H^{D,D^c})\bar P_{I_{\le q-\frac12}}(H^{D})}({\bf x},{\bf y})}\\
& \qquad 
\le \norm{{R_z^{D,D^c}\exp\pa{-r(H^{D,D^c}-E)^2}\bar P_{I_{\le q+\frac12}}(H^{D,D^c})}}
\le 2 \e^{-\frac r 4 }.
\ee
Moreover,
as  $ \abs{\x_{D^c}}= \abs{\y _{D^c}}\ge 1 $, we have $P_{I_{\le q+\frac12}}(H^{D,D^c})\bar P_{I_{\le q-\frac12}}(H^{D})=0$. Since 
\be
& \abs{\pa{F_{a}(H^{D,D^c}-E)\bar P_{I_{\le q-\frac12}}(H^{D})}({\bf x},{\bf y})}\\
& \qquad \le \abs{\pa{R^{D,D^c}_{z}P_{I_{\le q-\frac12}}(H^{D})}({\bf x},{\bf y})} +  \abs{\pa{F_{a}(H^{D,D^c}-E) P_{I_{\le q-\frac12}}(H^{D})}({\bf x},{\bf y})}, 
\ee 
we obtain the estimate
\be\label{eq:decena}
&\abs{G^{D,D^c}_{z}({\bf x},{\bf y})}\le \abs{\pa{R^{D,D^c}_{z}P_{I_{\le q-\frac12}}(H^{D})}({\bf x},{\bf y})} +     \abs{\pa{F_{a}(H^{D,D^c}-E)}({\bf x},{\bf y})}\\ &\hskip60pt 
+ \abs{\pa{F_{a}(H^{D,D^c}-E) P_{I_{\le q-\frac12}}(H^{D})}({\bf x},{\bf y})}+  2 \e^{-\frac r 4 }\\
&\hskip20pt \le \abs{\pa{R^{D,D^c}_{z}P_{I_{\le q-\frac12}}(H^{D})}({\bf x},{\bf y})} +  \abs{\pa{F_{a}(H^{D,D^c}-E) P_{I_{\le q-\frac12}}(H^{D})}({\bf x},{\bf y})}   \\&\hskip60pt  +C\pa{\e^{-\frac r 2} +\e^{-\frac r 4 }},
\ee
 where we used  \eqref{locest2} to get the last inequality.

We first estimate the second term in  the last line of \eq{eq:decena}. Given $\nu \in \sigma_{q-\frac12}(H^D)$, we get, using \eq{eigest},
\be
&\abs{\pa{F_{a}(H^{D,D^c}-E-\nu)\pi_{\nu}^D}(\x,\y)}\\
&\hspace{20pt}=\lceil q-\tfrac12\rceil\tfd  \abs{\pa{F_{a}(H^{D,D^c}-E-\nu) \what  R^{D}_{q-\frac12,\nu}\what Q^D_{\le \lceil  q-\frac12 \rceil} \pi_{\nu}^D}(\x,\y)}
\\ &\hspace{20pt}\le q\sum_{{\u_D}\in \cP_{N_D}(D)}\sum_{{\v_D}\in\cP_{{N_D},\lceil q-\frac12\rceil}(D)}\abs{F_{a}(H^{D,D^c}-E-\nu)({\bf x},\u_D \cup\y_{D^c})}\\&\hspace{7cm}\times\abs{\what  G^{D}_{ q-\frac12,\nu}(\u_D,\v_D)}\abs{\pi^D_{\nu}(\v_D,\y_D)}\\& \hspace{20pt}\le C_q\sum_{{\u_D}\in \cP_{N_D}(D)}\sum_{{\v_D}\in\cP_{{N_D},\lceil q-\frac12\rceil}(D)} \e^{-c d_H^{  \La } \pa{{\bf x},\u_D \cup\y_{D^c}}}\e^{-c_{q-\frac 1 2} d_1^D ({\u_D},{\v_D})}\abs{\pi^D_{\nu}(\v_D,\y_D)},
\ee
where  $N_D= \abs{\x_D}$ and we used   \eqref{eq:CTk} and \eqref{locest2} for the last  inequality.

It follows that 
\be\label{FHDDEP}
&\abs{\pa{F_{a}(H^{D,D^c}-E)P_{I_{\le q-\frac12}}(H^{D})}({\bf x},{\bf y})}\\ & \hspace{20pt} \le  
C_q {\sum_{{\u_D}\in \cP_{N_D}(D)}\sum_{{\v_D}\in\cP_{{N_D},\lceil q-\frac12\rceil}(D)} \e^{-c d_H^{  \La } {({\bf x},\u_D \cup\y_{D^c})}}\e^{-c_{q-\frac 1 2}d_1^D ({\u_D},{\v_D})}}\\&\hspace{6cm}\times\sum_{\nu \in \sigma_{q-\frac12}(H^D)}\abs{\pi^D_{\nu}(\v_D,\y_D)}\\ & \hspace{20pt}  \le  
C_q {\sum_{{\u_D}\in \cP_{N_D}(D)}\sum_{{\v_D}\in\cP_{{N_D},\lceil q-\frac12\rceil}(D)} \e^{-c d_H^{  \La } {({\bf x},\u_D \cup\y_{D^c})}}\e^{-c_{q-\frac 1 2}d_1^D ({\u_D},{\v_D})}} \cQ^D_{q-\frac 12}(\v_D,\y_D).
\ee

We have
\be\label{dHxyu}
d_H^{  \La } ({\bf x},\u_D \cup\y_{D^c}) +d_1^D ({\u_D},{\v_D})+d_H^D (\v_D,\y_D)  \ge d_H^{  \La }(\x,\y),
\ee
since
\be
d_H^D (\u_D,\y_D) =   d_H^\La (\u_D \cup\y_{D^c}, \y_D \cup \y_{D^c})=  d_H^\La (\u_D \cup\y_{D^c}, \y) ,
\ee
and hence
\be
&d_H^{  \La } ({\bf x},\u_D \cup\y_{D^c} )+d_1^D ({\u_D},{\v_D})+d_H^D (\v_D \cup \y_D)  \ge d_H^{  \La } ({\bf x},\u_D \cup\y_{D^c} ) + d_H^D (\u_D,\y_D)\\
& \quad =  d_H^{  \La } ({\bf x},\u_D \cup\y_{D^c} ) +   d_H^\La (\u_D \cup\y_{D^c}, \y) \ge    d_H^{  \La }(\x,\y).
\ee

Taking expectations in \eq{FHDDEP}, using  Lemma \ref{thm:eigencorweak} for $q-\frac 12$ (the hypotheses of the  Lemma  are satisfied for $q-\frac 12$ by the induction hypothesis), and using \eq{dHxyu},
we obtain the bound 
\be \label{FPxy}
&\E_\La \set{\abs{\pa{F_{a}(H^{D,D^c}-E) P_{I_{\le q-\frac12}}(H^{D})}({\bf x},{\bf y})}} \\& \quad  \le 
C_q {\sum_{{\u_D}\in \cP_{N_D}(D)}\sum_{{\v_D}\in\cP_{{N_D},\lceil q-\frac12\rceil}(D)} \e^{-c d_H^{  \La } {({\bf x},\u_D \cup\y_{D^c})}}\e^{-c_{q-\frac 1 2} d_1^D ({\u_D},{\v_D})}} \e^{-c_{q-\frac 12 } d_H^D (\v_D,\y_D)}\\
& \le  C_q \e^{-c_{q-\frac 12 }  d_H^{  \La } {({\bf x},\y)}} {\sum_{{\u_D}\in \cP_{N_D}(D)}\sum_{{\v_D}\in\cP_{{N_D},\lceil q-\frac12\rceil}(D)} \e^{-c d_H^{  \La } {({\bf x},\u_D \cup\y_{D^c})}}\e^{-c_{q-\frac 1 2}d_1^D ({\u_D},{\v_D})}} \e^{-c_{q-\frac 12 }  d_H^D (\v_D,\y_D)}\\
&\quad \le  C_q \e^{-c_{q-\frac 12 }  d_H^{  \La } {({\bf x},\y)}} \sum_{{\v_D}\in\cP_{{N_D},\lceil q-\frac12\rceil}(D)} 
 \pa{\sum_{{\u_D}\in \cP_{N_D}(D)}\e^{-c_{q-\frac 1 2}d_1^D ({\u_D},{\v_D})}} \e^{-c_{q-\frac 12 }  d_H^D (\v_D,\y_D)}\\
&\quad \le  C_q C_{q-\frac 1 2} \e^{-c_{q-\frac 12 }  d_H^{  \La } {({\bf x},\y)}} \sum_{{\v_D}\in\cP_{{N_D},\lceil q-\frac  12\rceil}(D)} 
  \e^{-c_{q-\frac 12 }  d_H^D (\v_D,\y_D)}\\
&\quad \le  C_q  N_D^{2\cl{q}}  \e^{-c_{q-\frac 12 }  d_H^{  \La } {({\bf x},\y)}}\le  C_q  \abs{D}^{2\cl{q}}  \e^{-c_{q-\frac 12 }  d_H^{  \La } {({\bf x},\y)}}, 
\ee  
where in the last two steps we used  \eq{eq:d1bn} and \eq{eq:dhbn}.  The use of the latter is justified  as $\y_D\in \cP_{{N_D},\lceil q\rceil}(D)$ since $ \y\in \cP_{{N},\lceil q\rceil}(\La)$ and $D$ is connected.

It remains to estimate the first term in \eq{eq:decena}.  We use  the decomposition 
 (recall that $D$ is assumed to be finite) 
\be
R_z^{D,D^c} =\sum_{\nu\in \sigma (H^{D})} R_{z-\nu}^{D^c}\otimes \pi_{\nu}^{D}  \qtx{on} \cH_\La= \cH_{D^c}\otimes \cH_{D},
\ee
Since $1\le \abs{\x _ D}= \abs{\y_D}\le N-1  $,
we have 
\be
\pa{R_z^{D,D^c}P_{I_{\le q-\frac12}}(H^{D})}({\bf x},{\bf y})=\sum_{\nu \in \sigma_{q-\frac12}(H^D)} G_{z-\nu}^{D^c}({\bf x}_{D^c},{\bf y}_{D^c})\, \pi^D_{\nu}({\bf x}_D,{\bf y}_D),
\ee 
so
\be\label{eq:fsm}
\abs{\pa{R_z^{D,D^c}P_{I_{\le q-\frac12}}(H^{D})}({\bf x},{\bf y})}^{s}\le\sum_{\nu \in \sigma_{q-\frac12}(H^D)}\abs{G_{z-\nu}^{D^c}({\bf x}_{D^c},{\bf y}_{D^c})}^{s}  \abs{ \pi^D_{\nu}({\bf x}_D,{\bf y}_D)}^{s}.
\ee
If  $z\in\mathds{H}_{q}$, we have  $z-\nu\in \mathds{H}_{\le q-\frac 12}$ for $  {\nu \in \sigma_{q-\frac12}(H^D)}$, so it follows from  the induction hypothesis that
\be
\sup_{\zeta \in\mathds{H}_{q- \frac 12}}\E_{D^c}\abs{G_{\zeta}^{D^c}({\bf x}_{D^c},{\bf y}_{D^c})}^{s}\le  C_{ q-\frac 12}  \e^{-c_{ q- \frac 12} d_H^{D^c} ({\bf x}_{D^c},{\bf y}_{D^c})}.
\ee
Thus, for  $ z\in\mathds{H}_{q}$, using  H\"older's inequality and  the deterministic estimate  \eq{trkH}.  we get
  \be\label{CNpinu}
&\E_{D^c} \set{\abs{\pa{R_z^{D,D^c}P_{I_{\le q-\frac12}}(H^{D})}({\bf x},{\bf y})}^s}\\& \qquad 
\le  C_{ q-1}\e^{-c_{ q-1} d_H^{D^c} ({\bf x}_{D^c},{\bf y}_{D^c})} \sum_{\nu \in \sigma_{q-\frac12}(H^D)}  \abs{ \pi^D_{\nu}({\bf x}_D,{\bf y}_D)}^{s}\\
& \qquad 
\le  C_{ q-1}\e^{-c_{ q-1} d_H^{D^c} ({\bf x}_{D^c},{\bf y}_{D^c})}\pa{ \sum_{\nu \in \sigma_{q-\frac12}(H^D)}  \abs{ \pi^D_{\nu}({\bf x}_D,{\bf y}_D)}}^{s}\pa{\tr \chi_{I_{ q-\frac 12}}(H^D)}^{1-s}\\
& \qquad 
= C_{ q-1}\e^{-c_{ q-1} d_H^{D^c} ({\bf x}_{D^c},{\bf y}_{D^c})}\pa{\tr \chi_{I_{ q-\frac 12}}(H^D)}^{1-s}\cQ^D_{q-\frac12}({\bf x}_D,{\bf y}_D)^s \\
& \qquad 
\le C_{ q-1}\pa{{ \cl{q-\tfrac 12} \abs{D}^{2\cl{q-\tfrac 12}} +1}}^{1-s}\e^{-c_{ q-1} d_H^{D^c} ({\bf x}_{D^c},{\bf y}_{D^c})}\cQ^D_{q-\frac12}({\bf x}_D,{\bf y}_D)^s \\
& \qquad 
\le C_{ q} \abs{D}^{2(1-s)\cl{q-\tfrac 12}} \e^{-c_{ q-1} d_H^{D^c} ({\bf x}_{D^c},{\bf y}_{D^c})}\cQ^D_{q-\frac12}({\bf x}_D,{\bf y}_D)^s .
\ee

It follows from H\"older's inequality, the induction hypothesis, and  Lemma \ref{thm:eigencorweak} for $q-\frac 12$
that
\be\label{CNpinu3}
\E_D \set{\pa{ \cQ^D_{q-\frac12}({\bf x}_D,{\bf y}_D)}^s}&
\le\pa{  \E_D \set{ \cQ^D_{q-\frac12}({\bf x}_D,{\bf y}_D)}}^s
\\   &  \le C_{q-\frac12}^s   \e^{-sc_{q-\frac12}d_H^D ({\bf x}_D,{\bf y}_D)}.
\ee
Combining  \eq{CNpinu} and \eq{CNpinu3} we get
\be\label{CNpinu98}
\E_{\La} \set{\abs{\pa{R_z^{D,D^c}P_{I_{\le q-\frac12}}(H^{D})}({\bf x},{\bf y})}^s}&\le   C_{ q}{\abs{D}}^{2\cl{q-\frac12}}\e^{-s c_{q-\frac12} \pa{ d_H^{D^c} ({\bf x}_{D^c},{\bf y}_{D^c})+d_H^D ({\bf x}_D,{\bf y}_D)} }\\
& \le    C_{ q}{\abs{D}}^{2\cl{q-\frac12}} \e^{-sc_{q-\frac12} { d_H^\La ({\bf x},{\bf y})} },
\ee
where we used 
\be
d_H^\La ({\bf x},{\bf y}) \le \max \pa{d_H^D ({\bf x}_D,{\bf y}_D),d_H^{D^c} ({\bf x}_{D^c},{\bf y}_{D^c})}.
\ee

It now follows from \eq{eq:decena}, \eq{FPxy}, and \eq{CNpinu98}, incorporating $s$ into the constants,
that
\be
E_\La\set{ \abs{G^{D,D^c}_{z}({\bf x},{\bf y})}^s} \le  C_{ q}\abs{D}^{2\cl{q}} \e^{-c_{q-\frac12} { d_H^\La ({\bf x},{\bf y})} },
\ee
and the lemma is proven.     
\end{proof}

\begin{lemma}\label{lem:spr}
Let $q\ \in  \frac 12 \N$, $1\le q $, and assume the induction hypothesis,  that is,  Theorem~\ref{cor:weakinf} is proven for $q-\frac 12$. 
Let  $\La \subset \Z$ and  $N\in \N$.  Then
\be\label{eq:spr}
\sup_{z\in\mathds{H}_q}\E_{\La}\set{\abs{G_z^\La ({\bf x},{\bf y}) }^{\ s}}\le C_q e^{-c_q d_H^\La(\x,\y)} \sqtx{for}\x,\y\in \cP_{N,\cl{q}}(\La),
\ee
\end{lemma}

\begin{proof}

Fix  $z\in\mathds{H}_q$ and $\x,\y\in \cP_{N,\cl{q}}(\La)$. We assume $d_H^\La ({\bf x},{\bf y})=  d_\La (x, {\bf y})$ for some $x\in \x$, with the other case being similar.

  We first assume $d_H^\La ({\bf x},{\bf y})>6\cl{q}N$.    In this case, we claim there exists $ r <d_H^\La ({\bf x},{\bf y})$, such that, setting 
 $D=[x]^\La_r $,   we have 
 \be\label{dLadH}
 d_\La (\x, \partial^\La D )\ge  \tfrac 1 {6\cl{q}} d_H^\La(\x,\y)-1 \qtx{and} d_\La (\y, \partial^\La D)\ge  \tfrac 1 {6\cl{q}} d_H^\La(\x,\y)-1. 
 \ee
Note that it follows that  $ \x_D \ne \emptyset$ and $\y_D=\emptyset$, which implies  $G_z^{D,D^c} ({\bf x},{\bf y})=0$.

The claim can be proven as follows. If $\cl{q}=1$,  or if $\x$ consists of one cluster, simply take $r= 3N$.  If  $\cl{q}\ge 2$, and   $\x$ consists of  $p$ clusters  where $2\le p \le \cl{q}$,  we must have  $N\ge 2$.
let  $b=\fl{\tfrac 1 {6\cl{q}}  d_\La (x, {\bf y})}>N-1$, and set 
 \be
 S_1=  [x]^\La_{b} \qtx{and} S_j=  [x]^\La_{jb}\setminus  [x]^\La_{(j-1) b } \sqtx{for} \quad j=2,3,,\ldots 6\cl{q}.
 \ee
 Since $\x\in \cP_{N,\cl{q}}(\La)$, $\x$ has at most $\cl{q}$ clusters of length $\le N-1$, so a cluster  can intersect at most two of the $S_{j}$'s (as $ b> N-1$), hence $\x $ can intersect at most $2\cl{q}$ of the $S_{j}$,  $j=2,3,\ldots 6\cl{q}$. It follows that there exists 
  ${j}_* \in \set{2,3,\ldots 6\cl{q}-2}$ such that 
 \be\label{eq:M}
 \x\cap \pa{S_{{j}_*}\cup S_{{j}_*+1}}=\emptyset,
 \ee  
Setting $r= {{j}_*}b$, we get \eq{dLadH}.

The resolvent identity and  $G_z^{D,D^c} ({\bf x},{\bf y})=0$ give
\be
G_z^\La ({\bf x},{\bf y})=\pa{R_z^{D,D^c} \Gamma^D  R_z^\La }({\bf x},{\bf y}),
\ee
so using   \eqref{eq:resmodl} and inserting partitions of identity, we get
\be\label{eq:par+}
&\abs{G_z^\La ({\bf x},{\bf y}) }\le C\sum_{{\bf u}\in\cP^D_N(\La)}  \abs{\what G_{z,q}^{D,D^c} ({\bf x},{\bf u})}\abs{ \pa{\Gamma^D R_z^\La} ({\bf u},{\bf y})}\\&\hspace{1cm}+C\sum_{{\bf u}\in \cP^D_N(\La)}\sum_{{\bf v}\in \cP_{N,\cl{q}}(\La)}     \abs{ G_z^{D,D^c} ({\bf x},{\bf v})} \abs{\what G_{z,q}^{D,D^c} ({\bf v},{\bf u})} \abs{ \pa{\Gamma^D R_z^\La} ({\bf u},{\bf y})} ,
\ee
where $\cP^D_N(\La)=\set{{\bf u}\in \cP_N(\La), \  {\bf u}_{\partial^\La D}\neq\emptyset}$.

For all $\u^\pr, \v^\pr \in \cP_N(\La)$ we have 
\be \label{eq:par+2}
&   \E\set{\abs{ \pa{\Gamma^D R_z^\La } (\u^\pr, \v^\pr)}^{s^\pr}}\le C_{s^\pr}  \qtx{for all } s^\pr  \in (0,1),  \\
&  \abs{\what G_{z,q}^{D,D^c} (\u^\pr, \v^\pr)} \le  C_q\e^{-c_q d_1^\La (\u^\pr, \v^\pr)} , 
\ee 
where the first bound follows from \eqref {eq:weak1-1}  since  $  \Gamma^D \phi_{\u^\pr} $ can be decomposed into a linear combination of at most 4 canonical basis vectors, and   the second is just  \eqref{eq:CTk}. 

We also have the inequality
 \be \label{GDDcxv}
\E_\La\set{\abs{ G_z^{D,D^c} ({\bf x},{\bf v})}^{s}}\le C_q\abs{D}^{C_q}  \e^{-c_q d^\La_H ({\bf x},{\bf v})} \qtx{for all}  \v\in \cP_{N,\cl{q}}(\La).
\ee 
If ${\bf x}_{D^c}\neq\emptyset$, this inequality follows from Lemma \ref {prop:decbn}. On the other hand, 
if ${\bf x}_{D^c}=\emptyset$,  $G_z^{D,D^c} ({\bf x},{\bf v})=0$  unless  ${\bf v}_{D^c}=\emptyset$, and in this case    $G_z^{D,D^c} ({\bf x},{\bf v})= G_z^{D} ({\bf x},{\bf v})$ and $d^\La_H ({\bf x},{\bf v})=d^D_H ({\bf x},{\bf v})$, and hence \eq{GDDcxv} follows from the hypothesis  \eqref{eq:FVC}.
Moreover, since $0<s <2s <\frac 23$, using the Riesz-Thorin Interpolation Theorem, it follows from \eq{GDDcxv}  and  \eqref {eq:weak1-1}   (with $s^\pr=\frac 23 $)  that
 \be \label{GDDcxv22}
\E_\La\set{\abs{ G_z^{D,D^c} ({\bf x},{\bf v})}^{2s}}\le C^\pr_q {\abs{D}}^{C^\pr_q}  \e^{-c^\pr_q d^\La_H ({\bf x},{\bf v})}.
\ee

From  \eq{eq:par+},\eq{eq:par+2},   \eq{GDDcxv22}, and $\abs{D} \le 2r +1$,  using also H\"older's inequality,  we get
\be\label{eq:almth55}
\sup_{z\in\mathds{H}_q}\E_\La\set{\abs{G_z^\La ({\bf x},{\bf y})}^{s}}&\le 
C_q\sum_{{\bf u}\in \cP^D_{N}} \e^{-c_q d_1^\La({\bf x},{\bf u})}\\&+
C_q r^{C_q}\sum_{{\bf u}\in \cP^D_N(\La)}\sum_{{\bf v}\in\cP_{N,\cl{q}}}  \e^{-c_q d_H^\La  ({\bf x},{\bf v})} \e^{-c_q d_1^\La({\bf u},{\bf v})}.
\ee

  Since $d_1^\La (\x,\u) \ge \tfrac 1 {6\cl{q}} d_H^\La(\x,\y)-1 $ for any ${\bf u}\in\cP^D_N(\La)$ by  \eq{dLadH}, we can bound 
\be
\sum_{{\bf u}\in \cP^D_{N}}\e^{-c_q  d_1^\La (\x,\u)}\le C_q e^{-\frac {c_q } {12\cl{q}} d_H^\La(\x,\y)}\sum_{{\bf u}\in \cP_{N}} \e^{-\frac {c_q} 2d_1^\La (\x,\u)}
\le C_q \e^{-c_q^\pr  d_H^\La(\x,\y)},
\ee  
where in the last step we used \eqref{eq:d1bn}. On the other hand, since it follows from \eq{dLadH} that 
\be
d_H^\La ({\bf x},{\bf v})+d_1^\La (\u,\v)\ge d_H^\La ({\bf x},{\bf u})\ge   \tfrac 1 {6\cl{q}} d_H^\La(\x,\y)-1
\ee 
for ${\bf u}\in \cP^D_N(\La)$, we can bound
\be
&\sum_{{\bf u}\in \cP^D_N(\La)}\sum_{{\bf v}\in\cP_{N,\cl{q}}}   \e^{-c_q d_H ({\bf x},{\bf v})} \e^{-c_q d_1^\La (\u,\v)}\le C_q r^{C^\pr_q}\e^{-c_q^\pr  d_H^\La(\x,\y)}\sum_{{\bf u}\in \cP^D_N(\La)}\sum_{{\bf v}\in\cP_{N,\cl{q}}}  \e^{-\frac  {c_q}2 d_H ({\bf x},{\bf v})} \e^{-\frac {c_q}2 d_1^\La (\u,\v)}\\& \quad  \le C_q  r^{C^\pr_q}N^{2\cl{q}}\e^{-c_q^{\prr } d_H^\La(\x,\y)} \le C_q  \pa{d_H^\La(\x,\y)}^{C_q}\e^{-c_q^{\prr } d_H^\La(\x,\y)}
\le  C_q  \e^{-c_qd_H^\La(\x,\y)},
\ee  
using \eqref{eq:dhbn} and \eqref{eq:d1bn}.  Using these bounds in \eqref{eq:almth55} yields \eq{eq:spr} if $d_H^\La ({\bf x},{\bf y})>6\cl{q}N$.  

It remains consider the case $d_H^\La ({\bf x},{\bf y})\le6\cl{q}N$.  To do so, we   will show that  
\be\label{eq:weakinfN}
 \sup_{z\in\mathds{H}_q}\E_\La\set{\abs{G^{\La}_{z}({\bf x},{\bf y})}^s}\le
  C_q\e^{-c_q d_1^\La (\x,\y)}+C_qe^{-c_qN}\mqtx{for all}  N \in \N,
\ee 
which yields, for $d_H^\La ({\bf x},{\bf y})\le6\cl{q}N$,
\be\label{eq:weakinfN987}
 \sup_{z\in\mathds{H}_q}\E_\La\set{\abs{G^{\La}_{z}({\bf x},{\bf y})}^s}\le 
 C_q\e^{-c_q d_1^\La (\x,\y)}+C_qe^{-c_qd_H^\La ({\bf x},{\bf y})}  \le C_qe^{-c_qd_H^\La ({\bf x},{\bf y})},
\ee 
establishing \eq{eq:spr}.  

To prove \eq{eq:weakinfN} we use a large deviation argument.
For $N\in \N$,   letting $S=[{\bf x}]^\La_N$, let $\cE^S_N$ be the event defined in \eq{eventENS}, so we have 
\eq{PV<}, and \eq{eq:CTkS} holds  on the  complementary event   $\pa{\cE_N^S}^c $.

 For $z\in \mathds{H}_{\clq}$  we also  have, using  H\"older's inequality  and the a-priori bound \eqref{eq:weak1-1},
  \be\label{eq:smset999}
\E_{\La}\set{\chi_{\cE_N^S}\abs{ G^{\La}_{z}({\bf x},{\bf y})}^s}\le\pa{ \P\set{\chi_{\cE_N^S} }}^{\frac 12}
\pa{\E_{\La}\set{\abs{ G^{\La}_{z}({\bf x},{\bf y})}^{2s}}}
\le C_q e^{-c_qN}.
\ee

On the  complementary event   $\pa{\cE_N^S}^c $ we use
\be\label{eq:smset48}
G^{\La}_{z}({\bf x},{\bf y})=G^{S,S^c}_{z}({\bf x},{\bf y})-\pa{R^{\La}_{z}\Gamma R^{S,S^c}_{z}}({\bf x},{\bf y}),
\ee
where $\Gamma=H^\La-H^{S,S^c}$. Since $\x \subset S$, we have  $G^{S,S^c}_{z}({\bf x},{\bf y})=0$ unless ${\bf y}\subset S$, in which case $G^{S,S^c}_{z}({\bf x},{\bf y})=G^{S}_{z}({\bf x},{\bf y})$. Thus \eqref{eq:CTkS} implies that in this case we have
\be\label{eq:b2a}
\sup_{z\in \mathds{H}_q} {\abs{ G^{S,S^c}_{z}({\bf x},{\bf y})}^s}\le C_q
e^{-c_qd_1^S (\x,\y)}\le C_q
e^{-c_qd_1^\La (\x,\y)}.
\ee  
On the other hand, setting $\cP_N^\partial(S) = \set{{\bf u}\in\cP_N(S): \ {\bf u}\cap \partial S\neq\emptyset}$, 
we have, for $\omega\notin \cE_N^S$ and $z\in \mathds{H}_q$, 
\be
\abs{\pa{R^{\La}_{z}\Gamma R^{S,S^c}_{z}}({\bf x},{\bf y})}^s&\le \sum_{{\bf u}\in\cP_N^\partial(S)} \abs{ \pa{R^{\La}_{z}\Gamma}({\bf x},{\bf u})}^s\abs{ G^{S,S^c}_{z}({\bf u},{\bf y})}^s\\&\le C_q \sum_{{\bf u}\in\cP_N^\partial(S)} e^{-c_q d_1^\La (\u,\y)}\abs{ \pa{R^{\La}_{z}\Gamma}({\bf x},{\bf u})}^s.
\ee 
It follows, using \eq{eq:par+2}, that  
\be
\sup_{z\in \mathds{H}_q}\E_{\La}\set{\chi_{\pa{\cE_N^S}^c}\abs{\pa{R^{\La}_{z}\Gamma R^{S,S^c}_{z}}({\bf x},{\bf y})}^s}\le C_q \sum_{{\bf u}\in\cP_N^\partial(S)} e^{-c_q d_1^\La (\u,\y)}\e^{-c_q d_H^\La ({\bf x},{\bf u})}. 
\ee
Since ${\bf u}\in\cP_N^\partial(S)$, we have $d_H^\La ({\bf x},{\bf u})\ge N$, and hence
\be\label{eq:b1a}
\sup_{z\in \mathds{H}_q}\E_{\La}\set{\chi_{\pa{\cE_N^S}^c}\abs{\pa{R^{\La}_{z}\Gamma R^{S,S^c}_{z}}({\bf x},{\bf y})}^s}\le C_q e^{-c_q N}\sum_{{\bf u}\in\cP_N^\partial(S)} e^{-c_q d_1^\La (\u,\y)}\le C_q \e^{-c_q N},
\ee
where we used \eqref{eq:d1bn} 
(recall $\y \in \cP_{N,\clq}(\La)$) to get the last inequality. 
Combining \eq{eq:smset48}, \eq{eq:b2a}, and \eqref{eq:b1a} we get      
\be\label{eq:b35a}
\sup_{z\in \mathds{H}_q}\E_{\La}\set{\chi_{\pa{\cE_N^S}^c} \abs{G^{\La}_{z}({\bf x},{\bf y})}^s}\le  C_q
e^{-c_qd_1^\La (\x,\y)} + C_q \e^{-c_q N}.
\ee

The estimate  \eqref{eq:weakinfN}  now follows from \eq{eq:smset999}  and    \eq{eq:b35a}.
\end{proof}
The first statement of the theorem (i.e., \eq{eq:weakinf}) now follows from Lemmas \ref{lem:spr} and \ref{lemweakstrong}.  The second statement (i.e., \eq{eq:smva+}) then follows from the  first statement  and Lemma~\ref{thm:eigencorweak}.
\end{proof}

\section{Proof of Theorem~\ref{thmdynloc}}\label{sec:main}

\begin{proof}  
We will show that the theorem can be derived from \cite[Theorem 4.1]{AW2}.
We start by reviewing the representation of the XXZ  quantum spin chain Hamiltonian by  a direct sum of discrete Schr\"odinger-like operators.

 As discussed in Section \ref{sec:feat},  given $\La \subset \Z$, we have the the Hilbert space decomposition $ \cH_\La= \bigoplus_{N=0}^{\abs{\La}} \cH_\La\up{N}$, where   $\cH_\Lambda\up{N}=\Ran {\chi_N(\mathcal N^\Lambda)}$.  We define
\be
\Z\up{N} = \set{(x_1,x_2,\ldots,x_N) \in \Z^N: \  x_1 < x_2<\ldots < x_N  } \qtx{and }\La\up{N}= \La^N \cap\Z\up{N}.
\ee
  Since $\H_\La\up{N}$ has the orthonormal basis $\Phi_\La\up{N}$, 
 identifying   $\x\in \cP_N(\La)$ with $(x_1,\ldots,x_N)\in \La\up{N}$ yields the identification of $ \cH_\La\up{N}$ with $\ell^2(\La\up{N})$.

Since  $H^\La, \bD^\Lambda, \cW^\Lambda, V^\La_\omega$ commute with the number of particles operator $\mathcal N^\Lambda$, they leave each $\cH_\La\up{N}$ invariant.  Let  $T_N^\La$ be the restriction of
$T^\La$ to  $\cH_\La\up{N}=\ell^2(\La\up{N})$, where  $T^\La= H^\La, \bD^\Lambda, \cW^\Lambda, V^\La_\omega$. 
We still have the decomposition given in \eq{bD}:
\be
H_N^\La= -\tfrac 1 {2\Delta} \bD^\Lambda_N +\cW^\Lambda_N +\lambda   V^\La_{N,\omega} \qtx{acting on} \ell^2(\La\up{N}),
\ee
where  $\bD^\Lambda_N $  is the adjacency operator on the graph $\La\up{N}$,  $\cW^\Lambda_N$ is a deterministic bounded potential, and $V^\La_{N,\omega}$ is a random potential.  In other words, $H_N^\La$
is a random Schr\"odinger operator on $\ell^2(\La\up{N})$.    For a fixed $N\in \N$, $H_N^\La$ satisfies all the 
hypothesis of the operators studied on \cite{AW2} except that it is  a Schr\"odinger operator on $\ell^2(\La\up{N})$,
not  on $\ell^2(\La^N)$.  This does not affect the analysis in \cite{AW2}, and all the results of \cite{AW2} hold for $H_N^\La$ for a fixed $N$.

Given $S\subset \Z$, we define the eigencorrelator   $\cQ^S_{N,q}({\bf x},{\bf y}) $ for $H_N^\La$ similarly  as we did for $H^S$  in Section \ref{secfinvol}. The hypothesis of the theorem can then be rewritten as:

\emph{Let $q\in \frac 12 \N^0$, and suppose that for all $D\subset \Z$ finite  and all  $N\in \N$ we have
 \be 
\E_{D}\set{\cQ_{N,q}^D ({\bf x},{\bf y})}\le  C_q \e^{-c_q  d_{H}^D({\bf x},{\bf y})}  \qtx{for all}\x,\y\in \ell^2(D\up{N}).
\ee
where the constants $C_q$ and $c_q$ are independent of $N$.}

We now apply \cite[Theorem 4.1]{AW2}  to $H_N^\La$ for all $N\in \N$. We
 obtain the conclusions of  the theorem for  $H_N^\Z$ for all $N\in \N$, with the constants independent of $N$ unless explicitly  stated.  It follows that the theorem holds as stated.
\end{proof}

\appendix

\section{Localization types and nomenclature}\label{sec:nom}
Localization is a very rich phenomenon that manifests itself in variety of ways.  As  discussed in Section \ref{secmodel}, for a single-particle systems one usually distinguishes  between  three types of localization: Spectral, eigenstate, and dynamical localization.  In this Appendix  we further describe these types, associated nomenclature, and the relationship between them.

\begin{enumerate}

\item {\it Spectral localization:} The spectrum of a random operator $H_\omega$ in a prescribed energy interval  $I$ is pure point, almost surely.  When $I=\R$, we say that $H_\omega$ is {\it completely spectrally localized}.

\item {\it Eigenstate localization:}  Consider   the Hilbert space $\ell^2(\La )$ where $\La$ is a subset of $\Z^d$,  the infinite system corresponding to $\La=\Z^d$, and finite systems corresponding to  bounded subsets $\La$.
 In  the mathematical literature,  a frequently used formulation is the semi-uniformly localized eigenvectors (SULE) form of  eigenstate localization: For a given energy interval  $I$ and almost any random configuration $\omega$, one can construct an orthonormal basis $\psi_{n,\omega}$ for the range of the spectral projection $P_I$ of $H_\omega$ onto $I$, such that, for each $n$ we can find a site $k_n\in\La$  so
\be
\abs{\psi_{n,\omega}(x)}\le C_{\omega}\langle k_n\rangle^p\e^{-m|x-k_n|},\quad x\in\La,\quad \langle k_n\rangle=\abs{k_n}+1,
\ee
 with parameters $p,m>0$ that do not depend on the choice of the configuration. That is, the normalized eigenfunction $\psi_{n,\omega}$ is exponentially confined near its {\it localization center} $k_n$, but the control over the confinement is only semi-uniform (it gets worse for localization centers further away from the origin).  Unfortunately, for the Anderson model, the primary model for studying single-particle localization phenomena, SULE localization cannot be upgraded to ULE (uniformly localized eigenvectors).   ULE does not occur for this model (and indeed for a broad class of random Schr\"odinger operators) \cite{DJLS}. 
 
 A related form of the eigenstate localization is exponential decay of the eigencorrelator (in expectation), already discussed in Section \ref{secfinvol}. Roughly speaking,  eigencorrelator decay implies SULE localization almost surely (see \cite{AW}).
 
 In physics, a popular metric for a measure of localization is the inverse participation ratio (IPR): If $\psi$ is a normalized eigenvector for $H_\omega$, the IPR for $\psi$ is given by $\sum_{x\in\La}\abs{\psi(x)}^4$. Using the Cauchy-Schwarz inequality, it is easy to see that 
 $1\ge\mathrm{IPR}(\psi)\ge \frac1{\abs{\La}}$, where the maximum is achieved when $\psi$ is a standard basis vector (that is maximally  localized), and the minimum is achieved when $\psi$ is uniformly spread on $\La$, i.e.,  $\psi(x)= \frac1{\sqrt{\abs{\La}}}$ for every $x\in\La$ (maximally delocalized state). We note that SULE implies that $\mathrm{IPR}(\psi_{n,\omega})\ge C>0$, where the constant $C$ is volume-independent and only depends (logarithmically) on the position of localization center for a given random configuration $\omega$.
 
If $H_\omega$ is completely spectrally localized, one can construct an orthonormal eigenbasis for $\ell^2(\La)$, that is, there exists a unitary operatoe $U_\omega$ that diagonalizes  $H_\omega$, i.e., $U_\omega^*H_\omega U_\omega$ is a diagonal matrix in the standard basis for $\ell^2(\La)$. It turns out  that for the Anderson model in the strong disorder regime, for any finite volume $\La\subset\Z^d$,  with high probability one can construct $U_\omega$ which is {\it semi-uniformly  quasi-local}, meaning that the matrix elements  $U_\omega(x,y)$ of $U_\omega$ satisfy the bound $|U_\omega(x,y)|\le  C_{\omega}\langle x\rangle^p\e^{-m|x-y|}$ for some $p,m>0$ that do not depend on the choice of the configuration. (This follows from the results in  \cite{EK}.)  The existence of such $U_\omega$ not only yields a SULE basis, but also allows to label these eigenfunctions according to the spatial position of their localization centers.  Such labeling is possible for the Anderson model due to the fact that one can show that with large probability the spectrum of $H_\omega$ is level spaced for  sufficiently regular distribution of the random potential.  Motivated by the concept of LIOM in the many-body context introduced in Section \ref{subsec:past}, we will refer to this form of the eigenstate localization as {\it semi-uniform LIOM localization}, As we already indicated, ULE does not occur for the Anderson model, so one can never upgrade a semi-uniform LIOM localization to a uniform LIOM localization in this context. 

\item {\it Dynamical localization:} In quantum mechanics, the dynamics can be given in either Schr\"odinger or Heisenberg pictures. For  single-particle systems, the Schr\"o\-dinger picture is  more common, whereas in the many-body context it is more natural to consider the Heisenberg picture. A typical object of interest in a single-particle system is the spread, due to the dynamics, of an initially localized wave packet. It can be characterized, for example, by {\it the transport exponent} $q$ defined by $\norm{X\e^{-itH_\omega}\delta_0}\sim \langle t\rangle^q$, where $\delta_0$ is the standard basis vector for $\ell^2(\La)$ located at the origin, and $X$ is a multiplication operator of the form $(X\psi)(x)=x\psi(x)$. The $q=0$ case is then associated with dynamical localization. We note that  complete spectral localization does not imply  dynamical localization, even if every eigenfunction is exponentially localized  \cite{DJLS1}.  SULE  and the eigencorrelator exponential decay are sufficient conditions to guarantee that $q=0$. 

 \end{enumerate}

The concepts in many-body localization quantify how much the influence of particle interaction affects the eigenfunctions and the dynamics. Since for a non-interacting system (considered perfectly many-body localized) the corresponding eigenvectors are product states, the measures of particle confinement in a single-particle system captured by various forms of localization are now replaced by measures of how far eigenvectors are from the product states and how fast information can propagate in these systems. Quantifying the former leads to the analogues of the eigenstate localization, and quantification of the latter produces the analogues of the dynamical localization.   As  mentioned in Section \ref{subsec:past},  dynamical localization in the many-body context does not follow from eigenstate or weak dynamical localization.

Dynamical localization can be expressed as {\it non-propagation of information}: For any observable $\mathcal O_{u}$ supported at site $u \in \Z$, $t\in\R$, and  $\ell\in\N$ there exists $m>0$ and an observable  $\mathcal O_{u,\ell,t}$  supported on $[u]_\ell$ such that 
\be\label{eq:info}
\norm{\tau_t^H(\mathcal O_{u})-\mathcal O_{u,\ell,t}}\le C\norm{\mathcal O_{u}}e^{-m\ell},
\ee
where $\tau_t^H(\mathcal O_{u}) = e^{itH}\mathcal O_{u}e^{-itH}$ is the Heisenberg evolution of $\mathcal O_{u}$.

Let $S=\sum_{n=1}^L\sigma_n^z$ with an arbitrary $L\in4\N$ and $H=S\sigma_0^z$ on $\cH_\Z$, then $H$ is diagonal  in the canonical  basis $\Phi^\Z$, and hence so is    $f(H)$ for any Borel function $f$.  It follows that $H$ satisfies a perfect weak dynamical localization condition in the sense that the right hand side of \eqref{eq:dloc} vanishes unless $\x=\y$. 

Letting $F(t)=\tau_t^H(\sigma_0^x)$ with $F(0)=\sigma_0^x$, we we have
\be
F^\pr(t)&=  -2 S  \tau_t^H(\sigma_0^y), \qtx{with}F^\pr(0)=-2 S  \sigma_0^y \ne 0,\\
F^{\pr\pr}(t)&=  -4 S^2  \tau_t^H(\sigma_0^x)=  -4 S^2  F(t)
\ee 
It follows that 
\be
F(t)=\cos(2St)\sigma_0^x -\sin(2St)\sigma_0^y .
\ee
Since $\e^{itS}=\prod_{n=1}^L\e^{it\sigma_n^z}$ and $e^{i\frac\pi2\sigma^z}=i\sigma^z$, we deduce that $F(\frac\pi2)=\prod_{n=1}^L\sigma_n^z\sigma_0^x$. So $\tau_t^H(\sigma_o^x)$ cannot satisfy \eq{eq:info} for say $t=\frac\pi2$.  Thus $H$ does not satisfy dynamical localization.

We conclude that weak dynamical localization of $H$  alone is not enough to show  dynamical localization, i.e., the non-spreading (or in fact even slow spreading) of information.

\section{Exponential sums}\label{sec:expsum}

 In this Appendix  we prove  bounds for the exponential sums  encountered throughout the paper.
They are used in the context of subsets $ \La \subset \Z$ and fixed number of particles $N\in\N$   in conjunction with the basic bounds  $d_1^\La (\x,\y)\ge d_1(\x,\y)=d^\Z_1(\x,\y)$ and
   $d^\La_H(\x,\y)\ge d_H(\x,\y)=d^\Z_H(\x,\y)= \abs{\x-\y}_1$ for   $\x,\y \in \cP_N(\La)$.  Note also that $\cP_{N,k}(\La) \subset \cP_{N,k}(\Z)$ for $k\in \N$.

\begin{lemma}\label{lem:expsum}
Let  $k,N\in \N$,  $k\le N$,   $\alpha>0$, and let
\beq
C_\alpha=(1-\e^{-\alpha})^{-1} \pa{ \prod_{n=1}^\infty (1- \e^{-\alpha n})^{-1}}^2.
\eeq 
  Then
\be\label{eq:d1bn43}
\sup_{\y\in \cP_{N}(\Z)} \sum_{{\x}\in\cP_{N,k}(\Z)} \e^{-\alpha \abs{\x-\y}_1}\le  C_{\alpha}^{k+1} ,
\ee 
and
\be\label{eq:d1bn}
\sup_{\x\in \cP_{N,k}(\Z)} \sum_{{\bf y}\in\cP_{N}(\Z)} \e^{-\alpha \abs{\x-\y}_1}\le C_{\alpha}^k ,
\ee  
\end{lemma}

\begin{proof}
The lemma is proven by adapting the argument of  \cite[Lemma B.2]{BeW},   who estimate the case $k=1$ of  \eqref{eq:d1bn43}.

Let $\x \in \cP_{N,k}(\Z)$, and  suppose that ${\bf x}$ has $m=m_\x$ clusters (where $m\in\set{1,\ldots,k}$). Then   ${\bf x}=(x_1, \ldots, x_N)$, where  $x_1< \ldots < x_N$, and let  $x_{j_1}< \ldots <x_{j_{2m}}$, where $j_1=1$, $j_{2m}=N$, be the end points for its $m$ clusters. (Note that  $m$ and  $j_2,,\ldots,j_{2m-1}$ are $\x$-dependent.) Given  $\y \in \cP_{N}(\Z)$  with $\y=(y_1, \ldots,y_N)$, 
where  $y_1<\ldots<y_N$, we set $t_i=y_i-x_i$,
 $i=1,\ldots,N$. Then the finite sequences  $\tau_{q}=(t_{j_{2q-1}},\ldots, t_{j_{2q}})$ are monotone non-decreasing for each  $q=1,\ldots,m$.   Let $\cT_{q}$ denote the collection of such monotone non-decreasing finite sequences $\tau_{q}$. Let $\cT\up{N}$ denote the collection of all monotone non-decreasing finite sequences $\tau\up{N}=(t_1,\ldots,t_N)$. 
 
 To prove \eqref{eq:d1bn43}, fix  $\y \in \cP_{N}(\Z)$.  Then each $\x \in \cP_{N,k}(\Z)$ is determined uniquely by the corresponding $\set{t_i}_{i=1}^N$, so we have 
 \be\label{appA19}
&\sum_{{\x}\in\cP_{N,k}(\Z)} \e^{-\alpha \abs{\x-\y}_1}=\sum_{{\x}\in\cP_{N,k}(\Z)}\prod_{q=1}^{m_\x}
\e^{-\alpha \sum_{j=2q-1}^{2q} \abs{x_j-y_j}}     \le 
\sum _{m=1}^k \sum_{ \tau_1\in \cT_{1}, \tau_2\in \cT_{1},\ldots, \tau_m\in\cT_{m} }     \prod_{q=1}^m \e^{-\alpha \sum_{t\in \tau_q} \abs{t}} 
\\  & \le  \sum _{m=1}^k \sum_{ \tau\up{N}_1,\tau\up{N}_2,\ldots,  \tau\up{N}_m\in\cT\up{N} }     \prod_{q=1}^m \e^{-\alpha \sum_{t\in \tau\up{N}_q} \abs{t}}  \le \sum _{m=1}^k  \pa{\sum_{\tau\up{N}\in\cT\up{N}} \e^{-\alpha \sum_{t\in \tau\up{N}} \abs{t}}}^m.
\ee
Given $\tau\up{N} \in \cT\up{N}$, since $\tau\up{N}$ is monotone non-decreasing   there exists an index $0\le p\le N+1$, such that $t_{j}< 0$ for $1\le  j\le  p$ and $t_{j}\ge 0$ for  $ {p}+1\le j\le  N$.
(Note that sets are allowed to be empty). Thus,
\be\label{appA29}
&\sum_{\tau \in \cT\up{N}}
\e^{-\alpha \sum_{j=1}^{N} \abs{t_{j}}}\\
& \  = \sum_{p=0}^{N+1}\pa{\sum_{t_1\le t_2\le \ldots \le t_p\le -1} \e^{\alpha (t_1 +t_2+\ldots +t_p)}}\pa{\sum_{0\le t_{p+1}\le t_{p+2}\le \ldots \le t_{N}} \e^{-\alpha (t_{p+1}+ t_{p+2}+\ldots +t_{N})}}\\
& \ = \sum_{p=0}^{N+1} \e^{-\alpha  p}\pa{\sum_{0\le t_1\le t_2\le \ldots \le t_p} \e^{-\alpha (t_1 +t_2+\ldots +t_p)}}\pa{\sum_{0\le t_{p+1}\le t_{p+2}\le \ldots \le t_{N}} \e^{-\alpha (t_{p+1}+ t_{p+2}+\ldots +t_{N})}}\\
&  \ \le  \pa{ \sum_{p=0}^{\infty} \e^{-\alpha  p}}  \pa{  \sum_{n=0}^\infty  P(n)    \e^{-\alpha  n} }^2
 =   (1-\e^{-\alpha})^{-1} \pa{ \prod_{n=1}^\infty (1- \e^{-\alpha n})^{-1}}^2=C_\alpha,
\ee
where $P(n) $ is the number of  integer partitions of $n$, and we used 
the formula for the generating function for $P(n) $.

It follows from \eq{appA19} and \eq{appA29} that 
\be
\sum_{{\x}\in\cP_{N,k}(\Z)} \e^{-\alpha \abs{\x-\y}_1} \le \sum_{m=1}^k  C_\alpha^m=  \tfrac {C_\alpha^{k+1}-C_\alpha}{C_\alpha-1}\le C_\alpha^{k+1},
\ee 
which yields \eq{eq:d1bn43}.

To establish \eqref{eq:d1bn}, we  modify  the above argument. We fix  $\x \in \cP_{N,k}(\Z)$, and note that every 
 ${\bf y} \in   \cP_{N}(\Z)$ is determined uniquely by the corresponding $\set{t_i}_{i=1}^N$, so we have 
\be\label{appA1}
\sum_{{\bf y}\in\cP_{N}(\Z)} \e^{-\alpha \abs{\x-\y}_1}&=\sum_{{\bf y}\in\cP_{N}(\Z)} \prod_{q=1}^m
\e^{-\alpha \sum_{j=2q-1}^{2q} \abs{x_j-y_j}} \le  \prod_{q=1}^m \sum_{\tau_q \in \cT_q}
\e^{-\alpha \sum_{j=j_{2q-1}}^{j_{2q}} \abs{t_{j}}}\\ &
=  \prod_{q=1}^m \sum_{\tau\up{n_q} \in \cT\up{n_q}}
\e^{-\alpha \sum_{j=1}^{n_q} \abs{t_{j}}} .
\ee
 where  
 $n_q=n_q(\x)= j_{2q}- j_{2q-1}$ for $q=1,2,\ldots,m$.

Let $n\in \N$, then
 \be\label{appA2}
&\sum_{\tau \in \cT\up{n}}
\e^{-\alpha \sum_{j=1}^{n} \abs{t_{j}}}\\
& \  = \sum_{p=0}^{N+1}\pa{\sum_{t_1\le t_2\le \ldots \le t_p\le -1} \e^{\alpha (t_1 +t_2+\ldots +t_p)}}\pa{\sum_{0\le t_{p+1}\le t_{p+2}\le \ldots \le t_{N}} \e^{-\alpha (t_{p+1}+ t_{p+2}+\ldots +t_{N})}}\\
& \ = \sum_{p=0}^{N+1} \e^{-\alpha  p}\pa{\sum_{0\le t_1\le t_2\le \ldots \le t_p} \e^{-\alpha (t_1 +t_2+\ldots +t_p)}}\pa{\sum_{0\le t_{p+1}\le t_{p+2}\le \ldots \le t_{N}} \e^{-\alpha (t_{p+1}+ t_{p+2}+\ldots +t_{N})}}\\
&  \ \le  \pa{ \sum_{p=0}^{\infty} \e^{-\alpha  p}}  \pa{  \sum_{n=0}^\infty  P(n)    \e^{-\alpha  n} }^2
=  C_\alpha,
\ee
 as in \eq{appA29}.

It follows from \eq{appA1} and \eq{appA2} that 
\be
\sum_{{\bf y}\in\cP_{N}(\Z)} \e^{-\alpha \abs{\x-\y}_1} \le C_\alpha^m,
\ee 
which yields \eq{eq:d1bn}.
\end{proof}

\begin{lemma}\label{lem:expsum2}
Let  $N\in \N$,  $k\in \N$, $k\le N$, and  $\alpha>0$.  Then
\be\label{eq:dhbn}
 \sup_{\x\in \cP_{N,k}(\Z)} \sum_{{\bf y}\in\cP_{N,k}(\Z)} \e^{-\alpha d_H({\bf x},{\bf y})}\le C_{\alpha,k} N^{2k}.
\ee
\end{lemma}

\begin{proof}
Fix $N \in \N$ and $k\in \N$, $k\le N$.
For $m \in \N$ let $\cP_N\up{m}(\Z)= \set{\x \in \cP_N(\Z), W_\x^\Z= m}$.    In addition, for $\x\in \cP_{N}(\Z)$, and  $r\in \N$ let 
 \be  
 \mathcal S_{{\bf x},r}& = \set{{\bf y}\in \cP_{N}(\Z), \ d_H({\bf x},{\bf y})=r},\quad \mathcal S_{{\bf x},r}\up{m} = \set{{\bf y}\in\cP_N\up{m}(\Z), \ d_H({\bf x},{\bf y})=r}\\
 \mathcal S_{{\bf x},r,k}& =  \bigcup_{m=1}^k \mathcal S_{{\bf x},r}\up{m}.
 \ee
 
Let now   $\x\in \cP_{N,k}(\Z)$.   We note that ${\bf y}\in\mathcal S_{{\bf x},r}$  implies ${\bf y}\subset  [{\bf x}]^\Z_r$. Since  
 $\abs{[{\bf x}]_r}\le N+ 2kr$, we deduce that  
 \be
 \abs{\mathcal S_{{\bf x},r}}\le \binom{ N+ 2kr}{N}\le (N+ 2kr)^{N}.
 \ee
As a consequence, for ${\bf x}\in  \cP_{N,k}(\Z)$ and  $\alpha >0$, we obtain the estimate
\be
\sum_{{\bf y}\in \cP_{N}(\Z)}e^{-\alpha d_H({\bf x},{\bf y})}&=1+\sum_{r=1}^\infty\abs{\mathcal S_{{\bf x},r}}\e^{-\alpha r}\\
&\le 1+{\sum_{r=1}^{\lfloor N/2{k}\rfloor}(2N)^{N}\e^{-\alpha r}+  \sum^\infty_{\fl{N/2{k}}+1}}  (4{k}r+1)^{N}\e^{-\alpha r}\le C_{\alpha,k}^N N^{N+1}.
\ee

We also have the  following bounds: 
\be
\abs{\mathcal S_{{\bf x},r}\up{m}}\le (N+2{k} r)^{2m},\quad \abs{\mathcal S_{{\bf x},r,k}}\le {   k} (N+2{k} r)^{2k}.
\ee
Clearly, the second bound follows immediately from the first one by summing over $m$. To obtain the first bound, we 
note that ${\bf y}\in\mathcal S_{{\bf x},r}\up{m}$  implies ${\bf y}\subset  [{\bf x}]^\Z_r$,  and  hence $\y$ is completely determined by the   $2m$ points in $[{\bf x}]^\Z_r$ that are the end points for its $m$ clusters. Since  
 $\abs{[{\bf x}]_r}\le N+  2{k} r$, we deduce that  
 \be
 \abs{\mathcal S_{{\bf x},r}\up{m}}\le \binom{ N+2{k} r}{2m}\le (N+2{k} r)^{2m}.
 \ee
 As a consequence, for ${\bf x}\in  \cP_{N,k}(\Z)$ and  $\alpha >0$, we obtain the estimate
\be
\sum_{{\bf y}\in \cP_{N,k}(\Z)}e^{-\alpha d_H({\bf x},{\bf y})}&=1+\sum_{r=1}^\infty\abs{\mathcal S_{{\bf x},r}}\e^{-\alpha r}\\
&\le 1+k\sum_{r=1}^{\lfloor N/2{k}\rfloor}(2N)^{2k}\e^{-\alpha r}+ k \sum^\infty_{ r=\fl{N/2{k}}+1} (4{k}r+1)^{2k}\e^{-\alpha r}\le C_{\alpha,k} N^{2k}.
\ee
\end{proof}

\section{Quasi-locality of the filter function} \label{app:quasil}  
 We fix $\La \subset \Z$ and consider the Hilbert space $\cH_\La$.  We  consider  disjoint subsets $K_1,K_2$ of $\La$,   and  let $H=H^{K_1} +H^{K_2}$ acting on $\cH_\La$.  We observe that the following holds:
 \begin{enumerate}
 \item For all   $K \subset {\La} $  and $K'=[K]_1^{K_1}$  we have   $[P_-^{K},H]P_+^{K'}=0$.
\item For all   $K \subset K_1 $,  connected in  $K_1$,   we have  $\norm{[P_-^{K},H]}\le  \gamma= \Delta^{-1}  $.
 \end{enumerate}
 
 We also observe that $H$ commutes with the total particle number operator, and its restriction $H_N$ to the N-particle sector $\cH_\La\up{N}$  is a well defined bounded operator  for each $N\in \N$.

We  use the following adaptation of 
 \cite[Lemma B.1]{EK22} which does not require $\La$ to be finite. 
 
 \begin{lemma} \label{lem:F}   For all $A\subsetneq B\subset {\La}$ with  $A, B$ finite,  $A\subset K_1$  connected in ${K_1}$,   we have
 \be \label{eq:loca1}
\norm{P_-^{A}\e^{itH} P_+^{B}}\le  \Delta^{-r}\frac{  \abs{t}^r}{r!} \sqtx{for all} t\in \R, \qtx{where} r=d_{\La} \pa{A,B^c}. 
\ee 
  \end{lemma} 
 
 \begin{proof}
  $H$ satisfies the input conditions (i) and (ii)  of  \cite[Lemma B.1]{EK22}, so its output (i.e., \eqref{eq:loca1}) is valid as well. 
 \end{proof}

\begin{theorem}\label{thmlocal} Given $t\in \R_+$ and $a\in\R$, 
let $F_{t,a}$ be the $C^\infty$  function on $\R$ given  in \eqref{defF}. 
Let   $S\subsetneq  T\subset \La$, $S,T$  finite, where $S\subset K_1$ is is connected in $K_1$,  and  let
$\ell= d_\Z\pa{S,T^c}-1$. Then for all  $t>1$,  $a\in \R$,  and  $E\in \R$ we have
 \be\label{locest}
\norm{P_-^{{S}}  F_{t,a}(H-E) P_+^{{T}}}\le  C    \pa{  \e^{-\frac 12\ell} +   \sqrt{t} \,\e^{-\frac { \Delta^2\ell^2}{100t}} },
\ee 
where all constants are $a$-independent. 

In particular,  taking  $t=  \frac {\Delta^2}{50} \ell $,    we have   ($\ell \in \N^0$)
\be\label{locest2}
\norm{P_-^{{S}} F_{  \frac {\Delta^2}{50}\ell,a}(H-E) P_+^{{T}}}  \le C\e^{- \frac 12 \ell} .
\ee 
\end{theorem}
We will refer to $F_{t,a}$ as a {\it filter function}.
\begin{proof} 

We introduce a introduce a function $F_{t,a,\eps} \in \cS(\R)$, where   $0<\eps$,
given by
    \be\label{defanf}
   F_{t,a,\eps} (x)&= \frac{\e^{-\eps x^2}-\e^{-t x^2}}{x-ia} \qtx{for} x \in \R \qtx{if} a\ne 0,\\
    F_{t,0,\eps} (x)&= \frac{\e^{-\eps x^2}-\e^{-t x^2}}{x} \qtx{for} x \in  \R\setminus \set{0} \qtx{and} F_{t,0} (0)=0.
   \ee 
Let  $\hat f$ denote the Fourier transform of the function $f$. 
We note that for $a>0$, the Fourier transform of $f_a(x)=\frac1{x-ia}$ exists as an $L^2$  function, and is given by $\hat f_a(\xi)=2i\pi\e^{a\xi}\chi_{(-\infty,0)}(\xi)$, whereas for $a=0$ it exists in a distributional sense, $\hat f_0(\xi)=-i\pi\sgn(\xi)$. We will only consider the more delicate case $a=0$, the argument for $a\neq0$ is very similar. 

 A standard calculation gives
\be
\hat F_{t,0,\eps}(\xi)
&= \tfrac 1 {\sqrt{2\pi}}\int_{-\infty}^\infty \pa{-i \sqrt{\tfrac \pi 2} \sgn (\xi-s) }\pa{\tfrac{1}{\sqrt{2\eps}} \e^{-\frac{s^2}{4 \eps}} -\tfrac{1}{\sqrt{2 t}} \e^{-\frac{s^2}{4 t}} } \, \d s\\ &
 = \tfrac i {2\sqrt{2}}  \pa{\int_{\xi}^{\infty} -  \int_{-\infty}^\xi } \pa{\tfrac{1}{\sqrt{\eps}} \e^{-\frac{s^2}{4 \eps}} -\tfrac{1}{\sqrt{ t}} \e^{-\frac{s^2}{4 t}} } \, \d s.
\ee 
Since  $F_{t,0,\eps}\in L^1$, by the Riemann--Lebesgue Lemma we have
\be
0= \lim_{\xi \to \infty} \hat F_{t,0,\eps}(\xi)= -\tfrac i {2\sqrt{2}}  \int_{-\infty}^\infty   \e^{i\lambda s}\pa{\tfrac{1}{\sqrt{\eps}} \e^{-\frac{s^2}{4 \eps}} -\tfrac{1}{\sqrt{ t}} \e^{-\frac{s^2}{4 t}} } \, \d s,
\ee
so it follows that
\be
\hat F_{t,0,\eps}(\xi)&= \tfrac i {\sqrt{2}}  \int_{\xi}^{\infty}\pa{\tfrac{1}{\sqrt{\eps}} \e^{-\frac{s^2}{4 \eps}} -\tfrac{1}{\sqrt{ t}} \e^{-\frac{s^2}{4 t}} } \, \d s\\ & \notag
=- \tfrac i {\sqrt{2}}  \int_{-\infty}^\xi   \pa{\tfrac{1}{\sqrt{\eps}} \e^{-\frac{s^2}{4 \eps}} -\tfrac{1}{\sqrt{ t}} \e^{-\frac{s^2}{4 t}}} \, \d s.
\ee

If $\xi >0$, we estimate
\be\label{Fxipi}
\abs{\hat F_{t,0,\eps}(\xi)}\le  \tfrac 1 {\sqrt{2}}   \int_{\xi}^{\infty}{\tfrac{1}{\sqrt\eps}\e^{-\frac{s^2}{4 \eps}} }  \, \d s+  \tfrac 1 {\sqrt{2}}  \int_{\xi}^{\infty}{\tfrac{1}{\sqrt t}\e^{-\frac{s^2}{4 t}} }\, \d s.
\ee
Using the Gaussian estimate
\be
\int_x^\infty \e^{-\frac {y^2}2}\, \d y \le \tfrac 1 x \e^{-\frac {x^2}2} \qtx{for} x>0,
\ee
and recalling $ \int_{0}^{\infty}\e^{-\frac{y^2}{2}}  \, \d y=\sqrt{\frac \pi  2}$, we conclude that 
\be
\int_x^\infty \e^{-\frac {y^2}2}\, \d y \le \sqrt{\tfrac {\pi\e } 2}  \, \e^{-\frac {x^2}2} \qtx{for all} x\ge 0.
\ee
(Note that $\int_x^\infty \e^{-\frac {y^2}2}\, \d y \le \e^{-\frac {x^2}2} $ for $x \ge 1$ and
$\e^{-\frac {x^2}2} \sqrt{\e} \ge 1$ for $x\in [0,1]$.)
Thus
\be \label{eq:erfc}
 \int_{\xi}^{\infty}\e^{-\frac {s^2}{4 b}}  \, \d s &= \sqrt{2b} \int_{\frac \xi{ \sqrt{2b}}}^{\infty}\e^{-\frac{y^2}2}  \, \d s \le \sqrt{\tfrac {\pi\e } 2}  \,  \sqrt{2b} \, \e^{-\frac{\xi^2}{4 b}}\\   &
 \le   \sqrt{\pi\e\, b} \, \e^{-\frac{\xi^2}{4 b}} \le  3\sqrt{b}  \, \e^{-\frac{\xi^2}{4 b}}\qtx{for} \xi\ge 0, \ b>0.
\ee 

It follows from \eq{Fxipi} and \eq{eq:erfc} that 
\be\label{eq:Ft1}
\abs{\hat F_{t,0,\eps}(\xi)}&   \le  \tfrac 3 {\sqrt{2}}\e^{-\frac{\xi^2}{4 \eps}} + \tfrac 3 {\sqrt{2}}\e^{-\frac{\xi^2}{4 t}}\le 5\e^{-\frac{\xi^2}{4 t}},
\ee
for all $\xi >0$, and since the same estimate can be established for $\xi<0$, for all $\xi \in \R$.  Moreover, the  same upper bound also holds for an arbitrary value of $a$.

 We can bound
\be\label{eq:2int}
\norm{P_-^A\ f(H)\, P_+^B}\le \int_{\mathcal R}\norm{P_-^A\ e^{itH}\, P_+^B}\abs{\hat f(t)}dt+\int_{\mathcal R^c}\abs{\hat f(t)}dt,
\ee
where $\mathcal R = [-R,R]$.   Using \eqref{eq:2int} with $R=c\ell$ and Lemma \ref{lem:F}, we have
\be
&\norm{P_- ^{S} f(H-E)P_+^{T}}
\le  C  \norm{\hat f}_\infty \frac{  \abs{\Delta^{-1} c\,  \ell}^\ell}{\ell!}+ \int_{\abs{t}> c\ell} \abs{ \hat f(t)}\, \d t.
\ee
Hence for $0<\eps<t$ and appropriately chosen value for $c$,   say $ c=\frac \Delta{5}$, 
  using Stirling's approximation and \eqref{eq:Ft1}, we get
\be\label{PFteps}
\norm{P_- ^{S} F_{t,0,\eps}(H-E)P_+^{T}}\le  C       \e^{-\frac 12 \ell}} + C\int_{\abs{\xi}> \frac \Delta{5} \ell} \e^{-\frac{\xi^2}{4 t}} \, \d \xi \le C \e^{- \frac 12 \ell}+ C\sqrt{t} \, 
 \e^{-\frac {\Delta^2\ell^2 }{100 t}    .          
\ee
Using $\abs{(F_{t,0}-F_{t,0,\eps})(x)}\le \eps \abs{x}$, \eq{PFteps},   restricting to   the N-particle sector $\cH_\La\up{N}$, and recalling that $H_N$ is a bounded operator,  we get 
\be\label{locest6666}
\norm{P_-^{{S}} F_{t,0}(H_{N}-E) P_+^{{T}}}& \le \norm{P_-^{{S}} F_{t,0,\eps}(H_{N}-E) P_+^{{T}}} + \norm{ (F_{t,0}-F_{t,0,\eps})(H_{N}-E)}\\  &
\le  C    \pa{  \e^{- \frac 12\ell} +   \sqrt{t} \,\e^{-\frac {\Delta^2\ell^2 }{100 t}   }}
+ \eps \norm{H_{N}-E},
\ee  
 where $C$ is $N$-independent. 
Letting $\eps \to 0$ we get
\be\label{locest'}
\norm{P_-^{{S}}  F_{t,0}(H_N-E) P_+^{{T}}}\le  C    \pa{  \e^{-\frac 12\ell} +   \sqrt{t} \,  \e^{-\frac {\Delta^2\ell^2 }{100 t}   }  }\qtx{for all} N\in \N.
\ee 
The desired estimate \eq{locest} follows.
\end{proof}

\section*{Declarations}

\subsection*{\qquad Data availability}
We do not analyze or generate any datasets, because our work proceeds within a theoretical and mathematical approach. One can obtain the relevant materials from the references below.

\subsection *{\qquad  Funding and/or Conflicts of interests/Competing interests} \

Alexander Elgart was  supported in part by the NSF under grant DMS-2307093.

The authors have no relevant financial or non-financial interests to disclose.

The authors have no competing interests to declare that are relevant to the content of this article.

\printbibliography

\end{document}